\documentclass[a4paper,UKenglish,cleveref, autoref, thm-restate]{lipics-v2021}

\hideLIPIcs  

\title{Dicey Games: Shared Sources of Randomness in Distributed Systems} 

\author{L\'eonard Brice}{Institute of Science and Technology Austria, Austria \and \url{https://lnrdbrice.github.io/} }{leonard.brice@ista.ac.at}{https://orcid.org/0000-0001-7748-7716}{Part of this work was realised when this author was an FNRS aspirant at Université libre de Bruxelles.}

\author{Thomas A. Henzinger}{Institute of Science and Technology Austria, Austria \and \url{https://pub.ista.ac.at/~tah/}}{tah@ista.ac.at}{0000-0001-6077-7514}{}

\author{K. S. Thejaswini}{Universit\'e libre de Bruxelles, Belgium \and \url{https://thejaswiniraghavan.github.io/}}{thejaswini.raghavan@ulb.be}{0000-0001-6077-7514}{Part of this work was realised when this author was employed at IST Austria.}

\authorrunning{L. Brice, T. A. Henzinger, K. S. Thejaswini} 

\Copyright{L\`eonard Brice, Thomas A. Henzinger, K. S. Thejaswini} 



\relatedversion{https://arxiv.org/abs/2601.18303} 

\funding{This work was supported in part by the ERC-2020-AdG 101020093 (VAMOS).}

\acknowledgements{We thank all our colleagues who took the time to hear our puzzle and wasted several hours of their research time in pursuit of the optimal bounds for the 3-player matching pennies problem.}

\nolinenumbers 

\EventEditors{Claudia Faggian and Joost-Pieter Katoen}
\EventNoEds{2}
\EventLongTitle{41st Annual Symposium on Logic in Computer Science (LICS 2026)}
\EventShortTitle{LICS 2026}
\EventAcronym{LICS}
\EventYear{2026}
\EventDate{July 20--23, 2026}
\EventLocation{Lisbon, Portugal}
\EventLogo{}
\SeriesVolume{380}
\ArticleNo{30}

\usepackage{orcidlink}
\usepackage{tikz}
\usepackage{thm-restate}
\usepackage{hyperref}
\usepackage[capitalise]{cleveref}
\usepackage{complexity}
\usepackage{caption}
\usepackage{subcaption}
\usepackage{graphicx}
\usepackage{stmaryrd}
\usepackage{pgfplots}
\pgfplotsset{compat=1.18}
\usepackage{tikz-3dplot}
\usepackage[many]{tcolorbox}
\usepackage{csquotes}

\let\oldint\int
\renewcommand{\int}{\oldint\!\!}

\newtheorem*{question}{Question}
\newtheorem{problem}{Problem}

\renewcommand{\epsilon}{\varepsilon}
\renewcommand{\phi}{\varphi}
\renewcommand{\l}{\ell}

\renewcommand{\Game}{\mathcal{G}}
\newcommand{\Pack}{\mathcal{P}}
\newcommand{\Struct}{\mathcal{S}}
\newcommand{\DGame}{\mathcal{D}}

\newcommand{\Players}{\Pi}

\newcommand{\Actions}{A}

\newcommand{\Dice}{\Delta}
\newcommand{\Universe}{\Omega}

\newcommand{\Die}{D}
\newcommand{\acc}{\mathsf{acc}}
\newcommand{\Acc}{\mathsf{Acc}}

\newtcolorbox{boxB}{
    fontupper = \color{black}, 
    boxrule = 1.5pt,
    colframe = black,
    rounded corners,
    arc = 5pt   
}

\newcommand{\payoff}{\mu}

\newcommand{\player}{p}

\newcommand{\action}{a}

\newcommand{\roll}{r}
\newcommand{\strat}{\sigma}
\newcommand{\threshold}{t}

\newcommand{\everyplayer}{{\player \in \Players \cup \{\opp\}}}

\newcommand{\broll}{{\bar{\roll}}}
\newcommand{\bstrat}{{\bar{\strat}}}
\newcommand{\baction}{{\bar{\action}}}

\newcommand{\Conv}{\mathsf{Conv}}
\newcommand{\rand}{\mathsf{rand}}
\newcommand{\Active}{\mathsf{Active}}

\newcommand{\hx}{\hat{x}}
\newcommand{\hy}{\hat{y}}   

\newcommand{\bv}{\bar{v}}

\newcommand{\bzero}{\bar{0}}
\newcommand{\bx}{\bar{x}}
\newcommand{\by}{\bar{y}}
\newcommand{\bz}{\bar{z}}
\newcommand{\bbx}{\bar{\bar{x}}}

\newcommand{\bbaction}{{\bar{\baction}}}
\newcommand{\bblambda}{{\overline{\overline{\phantom{a}}} \hspace{-1.8mm}\lambda}}

\newcommand{\ada}{\mathfrak{A}}
\newcommand{\bert}{\mathfrak{B}}
\newcommand{\clau}{\mathfrak{C}}
\newcommand{\dev}{\mathfrak{D}}
\newcommand{\opp}{\dev}

\renewcommand{\H}{\mathsf{H}}
\newcommand{\T}{\mathsf{T}}

\newcommand{\prob}{\mathbb{P}}

\newcommand{\Nb}{\mathbb{N}}
\newcommand{\Rb}{\mathbb{R}}
\newcommand{\Eb}{\mathbb{E}}
\newcommand{\Qb}{\mathbb{Q}}
\newcommand{\Zb}{\mathbb{Z}}

  \newclass{\FNEXP}{FNEXP}
 \newclass{\NEXPSPACE}{NEXPSPACE}
 \newclass{\IIEXP}{2EXP}
 \newclass{\FIIEXP}{F2EXP}

\newcommand{\DQBF}{\mathrm{DQBF}}

 \usepackage{todonotes}

\usepackage{tikz}
\usetikzlibrary{patterns, decorations.pathmorphing}

\usepackage[framemethod=TikZ]{mdframed}
\newtheorem{sublemma}[theorem]{Sublemma}
\crefname{slem}{Sublemma}{Sublemmas}

\newcommand{\val}{\mathsf{val}}

\newcounter{quest}[section]

\newcounter{conjec}[section]
\newenvironment{conjec}[1][]{%
    \stepcounter{conjec}%
    \ifstrempty{#1}%
    {\mdfsetup{%
        frametitle={%
            \tikz[baseline=(current bounding box.east),outer sep=0pt]
            \node[anchor=east,rectangle,fill=red!20]
            {\strut};}}
    }%
    {\mdfsetup{%
        frametitle={%
            \tikz[baseline=(current bounding box.east),outer sep=0pt]
            \node[anchor=east,rectangle,fill=red!20]
            {\strut #1};}}%
    }
    \mdfsetup{
        innertopmargin=0pt,
        linecolor=red!20,
        linewidth=2pt,
        topline=true,
        frametitleaboveskip=\dimexpr-\ht\strutbox\relax
    }
    \begin{mdframed}[]\relax%
}{\end{mdframed}}

\begin{document}

\maketitle

\begin{abstract}
 Consider a 4-player version of Matching Pennies where a team of three players competes against the Devil. Each player simultaneously says ``Heads'' or ``Tails''.  The team wins if all four choices match; otherwise the Devil wins. If all team players randomise independently, they win with probability $1/8$; if all players share a common source of randomness, they win with probability $1/2$.
What happens when \textit{each pair} of team players shares a source of randomness?
Can the team do better than win with probability 1/4? The surprising (and nontrivial) answer is yes!

We introduce Dicey Games, a formal framework motivated by the study of distributed systems with shared sources of randomness (of which the above example is a specific instance).
We characterise the existence, representation and computational complexity of optimal strategies in Dicey Games, and we study the problem of allocating limited sources of randomness optimally within a team.
\end{abstract}

\section{Introduction}

Concurrent games, and especially two-player zero-sum concurrent games, provide natural models of limited information, distributed decision-making, and simultaneous interaction. 

One example of concurrent games is \emph{matching pennies}. Two players, whom we refer to as Ada and the Devil, simultaneously declare either ``Heads'' or ``Tails''. Ada wins if both her and the Devil's declaration agree, that is, they both say Heads or they both say Tails, while the Devil prefers disagreement.
This is a zero-sum concurrent game, and the best possible strategy for Ada ensures her victory with probability $1/2$.
For two-player zero sum concurrent games, the
\emph{value} of the game is well defined, since the $\max\min$
value of the game (the largest payoff that Ada can guarantee) is equal to the $\min\max$ value (the smallest payoff the Devil can force).

Now consider a version of this game with one additional player: Bertrand. 
Similar to the previous examples, each player simultaneously says one of either ``Heads'' or ``Tails''. Ada and Bertrand now form a team, and win if all three players declare the same action; the Devil wins otherwise.
Ada and Bertrand may communicate beforehand to agree on a strategy, but they do not observe each other’s private random choices at play time. 
In such a game, the  $\max\min$ value is in general not equal to the $\min\max$ value.

Indeed, if the Devil commits first to a given strategy, then Ada and Bertrand can ensure 
victory
with probability at least $1/2$ (by both deterministically choosing the action that the Devil chooses with the largest probability). The $\min\max$ value of this game is therefore $1/2$. Conversely, assume Ada and Bertrand fix their strategies first. Ada picks ``Heads'' and ``Tails'' with probability $p_A$ and $1-p_A$, respectively, while Bertrand picks ``Heads'' and ``Tails'' with probability $p_B$ and $1-p_B$, respectively.
Then, the Devil can deterministically choose either Heads or Tails to ensure that Ada and Bertrand win with probability at most $\min\{p_Ap_B, (1-p_A)(1-p_B)\}$.
This quantity is maximised when $p_A=p_B = 1/2$, ensuring that the $\max\min$ value is $1/4$.
In this paper, we focus on the $\max\min$ value. 

A natural extension introduces a third team member, Claude, who joins Ada and Bertrand in attempting to match pennies with the  Devil.
If each player uses an individual source of randomness,
the optimal strategy for each player is to choose independently and uniformly, giving a $\max\min$ value of $1/8$.

Shared randomness, however, changes the picture.
We define a shared source of randomness (or \emph{die}) as a random variable ranging uniformly over the interval $[0,1]$, and to which only a given set of players (always excluding the Devil) have access (note that every continuous distribution can be mapped to the uniform one).
Suppose first that  Ada and Bertrand share a die.
Then, the team can, again, ensure victory with probability $1/4$, since Ada and Bertrand's outcomes can always be ensured to match.
Similarly, if all three team players share a die, the situation collapses to the two-player matching-pennies case (Ada vs. Devil) and the $\max\min$ value becomes $1/2$.

We now pose an illustrative question that motivates the broader class of games that we call \emph{dicey games}, addressed in this work.


\vspace{3mm}
\begin{conjec}[Triangular matching pennies with the Devil
]
What is the optimal strategy---and corresponding $\max\min$ value---for achieving an all-Heads or all-Tails outcome in the matching-pennies game played by Ada, Bertrand, and Claude against the Devil, under the assumption that Ada and Bertrand share a die, Bertrand and Claude share another, and Claude and Ada share a third?
\end{conjec}
\vspace{0.2cm}
Suppose, to begin with, that Ada and Bertrand rely solely on their shared die,
while Claude bases his action only on one of the other two dice available to him. In this restricted setting, we have already observed that the team can guarantee a success probability of at most $1/4$.
At this point we invite the reader to reflect: can the trio do better than $1/4$ when all three pairwise shared sources of randomness are available simultaneously? 

Since each player has access to two independent shared random variables, a strategy for any given player may be viewed as a measurable function $[0,1]^2 \to \{\H,\T\}$. Equivalently, one may visualise such a strategy as a partition of the unit square into two regions corresponding to ``Heads'' and ``Tails''.

A natural attempt a reader might consider is the following: declare Heads whenever the sum of the two shared random values is at most~$1$, and Tails otherwise. This strategy is depicted in \cref{fig:tetrahedron}. The probability that all three players simultaneously declare Heads (or all declare Tails) then corresponds to the volume of the red (resp blue) tetrahedral region shown in \cref{fig:TetrahedralHHH}, which is exactly $1/6$. Somewhat disappointingly, this is less than the value $1/4$
found earlier---but still better than if they randomise independently.

\begin{figure}
    \centering
    \begin{subfigure}[t]{0.25\textwidth}
    \centering
	\begin{tikzpicture}[scale=0.18]  
        \draw[->] (0,0) -- (0,10);
        \draw[->] (0,0) -- (10,0);

        \draw (10,0) node[below right] {$x$};
    	\draw (0,10) node[above left] {$y$};

        \fill[red, opacity=0.4] (0,0) -- (9,0) -- (0,9) -- cycle;
        \fill[blue, opacity=0.4] (9,0) -- (0,9) -- (9,9) -- cycle;
        \draw[red] (2.5,2.5) node {$\H$};
        \draw[blue!70!black] (6.5,6.5) node {$\T$};
    \end{tikzpicture}
    \caption{Ada's strategy}
    \label{fig:TetrahedralAda}
    \end{subfigure}
       \begin{subfigure}[t]{0.25\textwidth}
    \centering
	\begin{tikzpicture}[scale=0.18]  
        \draw[->] (0,0) -- (0,10);
        \draw[->] (0,0) -- (10,0);

        \draw (10,0) node[below right] {$y$};
    	\draw (0,10) node[above left] {$z$};

        \fill[red, opacity=0.4] (0,0) -- (9,0) -- (0,9) -- cycle;
        \fill[blue, opacity=0.4] (9,0) -- (0,9) -- (9,9) -- cycle;
        \draw[red] (2.5,2.5) node {$\H$};
        \draw[blue!70!black] (6.5,6.5) node {$\T$};
    \end{tikzpicture}
    \caption{Bertrand's strategy}
    \label{fig:TetrahedralBertrand}
    \end{subfigure}
       \begin{subfigure}[t]{0.25\textwidth}
    \centering
	\begin{tikzpicture}[scale=0.18]  
        \draw[->] (0,0) -- (0,10);
        \draw[->] (0,0) -- (10,0);

        \draw (10,0) node[below right] {$z$};
    	\draw (0,10) node[above left] {$x$};

        \fill[red, opacity=0.4] (0,0) -- (9,0) -- (0,9) -- cycle;
        \fill[blue, opacity=0.4] (9,0) -- (0,9) -- (9,9) -- cycle;
        \draw[red] (2.5,2.5) node {$\H$};
        \draw[blue!70!black] (6.5,6.5) node {$\T$};
    \end{tikzpicture}
    \caption{Claude's strategy}
    \label{fig:TetrahedralClaude}
    \end{subfigure}
    \begin{subfigure}[t]{0.35\textwidth}
    \centering
    \vspace{0pt}
	\begin{tikzpicture}[scale=1.8, line join=round]  
\newcommand{\tetrahe}[5]{%
    \fill[#5, opacity=0.4] #1 -- #2 -- #3 -- cycle;
    \fill[#5, opacity=0.4] #1 -- #2 -- #4 -- cycle;
    \fill[#5, opacity=0.4] #1 -- #3 -- #4 -- cycle;
    \fill[#5, opacity=0.4] #2 -- #3 -- #4 -- cycle;

    \draw[thick] #1 -- #2 -- #3 -- cycle;
    \draw[thick] #1 -- #4;
    \draw[thick] #2 -- #4;
    \draw[thick] #3 -- #4;
}

        \draw[->] (0,0,0) -- (1.2,0,0) node[right] {$x$};
        \draw[->] (0,0,0) -- (0,1.2,0) node[above] {$y$};
        \draw[->] (0,0,0) -- (0,0,1.2) node[left] {$z$};

 \draw[dashed] (0,0,0) -- (1,0,0) -- (1,1,0) -- (0,1,0) -- cycle;

 \draw[dashed] (0,0,1) -- (1,0,1) -- (1,1,1) -- (0,1,1) -- cycle;

\draw[dashed] (0,0,0) -- (0,0,1);
\draw[dashed] (1,0,0) -- (1,0,1);
\draw[dashed] (1,1,0) -- (1,1,1);
\draw[dashed] (0,1,0) -- (0,1,1);

\tetrahe{(0,0,0)}{(1,0,0)}{(0,1,0)}{(0,0,1)}{red}

\tetrahe{(1,1,1)}{(0,1,1)}{(1,0,1)}{(1,1,0)}{blue}
    \end{tikzpicture}
    \caption{The sets $\H\H\H$ and $\T\T\T$}
    \label{fig:TetrahedralHHH}
    \end{subfigure}
    \caption{A strategy to win with probability 1/6}
    \label{fig:tetrahedron}
\end{figure}

Readers who prefer to discover the optimal construction on their own may wish to pause here, as  
we will now reveal the solution.

The optimal collective strategy turns out to possess a simple geometric form.  Each player declares Heads if the value of both their dice are at least $\alpha$, 
for some fixed value
$\alpha \in [0,1]$, and Tails otherwise. The probability of an all-Heads (resp. all-Tails) outcome in such a case is equal to the volume of the red (resp. blue) solid in \cref{fig:optimalHHH}.
The value of $\alpha$ that optimises the $\max\min$ value is the one that makes that the volume of the red cube is exactly equal to the volume of the blue solid. Solving the cubic polynomial that captures that equality, we can conclude that the team’s winning  probability under this strategy is approximately $0.2781$.
Proving that this is optimal is however more challenging. 

\subsection{Our results}
In this work, we investigate a class of concurrent games called \emph{dicey games}. These are matrix games played by a team of players against a single opponent, the Devil, where the team may coordinate its actions through a finite collection of shared random sources, or \emph{dice}. Each die is observable to a designated subset of the team, and the pattern of observability determines the correlation structure available to the players. As in standard concurrent games, the payoff to the team is specified for every joint action of the team and the Devil. Formally, the dice structure is given by a set of dice, and a map assigning to each die the subset of players who may inspect its outcome. For example, in the triangular matching pennies game, we have a set of dice $\{D_1,D_2,D_3\}$, together with the assignment $D_1\mapsto \{\text{Ada, Bertrand}\}$, $D_2\mapsto\{\text{Bertrand,  Claude}\}$, $D_3\mapsto \{\text{Claude, Ada}\}$.
 Our work is motivated by the central question: 
 \begin{question}\label{question1}
Which collective strategies are sufficient to reach the $\max\min$ value in dicey games?
 \end{question}
In this framework, we prove two main results.

Our first result (\Cref{thm:itsOKtobestraight}) is a normal form for an optimal collective strategy. We prove that if  the team 
has a collective strategy ensuring expected payoff at least $t$, then it has a \emph{$k$-grid} one, where $k$ is the number of actions available to the Devil. 
To explain the notion of $k$-grid strategy, first recall that a strategy of a player with access to $\ell$ different dice is a map from the space $[0,1]^\ell$ to the set of actions available to that player. Equivalently, it can be seen as a labelled partition of the hypercube $[0,1]^\ell$, where each part 
is labelled by an action.
Similarly, if $d$ is the total number of dice, a collective strategy induces a partition of the unit cube $[0,1]^d$, labelled with tuples of actions.
Geometrically, a $k$-grid collective strategy partitions the 
unit cube into 
sub-cubes,
obtained by cutting each coordinate axis into $k$ intervals.

Describing such a collective strategy requires specifying the $k-1$ thresholds at which the axis is cut, for each axis.
It is then natural to wonder how these values can be represented: already in triangular matching pennies, both the optimal strategy and the resulting value bring us to the realm of irrational numbers. 

Our second result (\Cref{thm:exponentialbitsize}) shows that in every dicey game, there is an optimal collective strategy, with a finite representation whose bit-size grows exponentially (but not more) with the size of the game.
Establishing this bound  is non-trivial, requiring a novel blend of techniques from computational algebraic geometry (results by Collins~\cite{Col74} and Basu, Pollack, and Roy~\cite{BPR94}) and non-linear optimisation (Fritz John conditions~\cite{John48}). The resulting machinery might prove useful for attacking similar problems in the future.

Finally, we use those theorems to prove several complexity results (\Cref{thm:valuecomputation,thm:threshold_complexity,thm:allocate_complexity}) about natural problems related to dicey games.
Those include (1) deciding whether there exists some collective strategy to achieve at least a given threshold,  (2) computing the $\max\min$ value of a given dicey game, and (3) the corresponding threshold and value problems when the shared randomness structure is not fixed in advance and the team must decide how to allocate a fixed collection of dice among the players. 
The threshold problems require insights that we develop in \cref{sec:example}, while the computational versions require further mathematical work developed in \cref{sec:optimal}.

\begin{figure}[h]
    \centering
    \begin{subfigure}[t]{0.22\textwidth}
    \centering
	\begin{tikzpicture}[scale=0.18]  
        \draw[->] (0,0) -- (0,10);
        \draw[->] (0,0) -- (10,0);

        \draw (10,0) node[below right] {$x$};
    	\draw (0,10) node[above left] {$y$};

        \fill[red, opacity=0.4] (3,3) -- (3,9) -- (9,9) -- (9,3) -- cycle;
        \fill[blue, opacity=0.4] (0,0) -- (9,0) -- (9,3) -- (3,3) -- (3,9) -- (0,9) -- cycle;
        \draw[red] (6,6) node {$\H$};
        \draw[blue!70!black] (1.5,1.5) node {$\T$};
    \end{tikzpicture}
    \caption{Ada's strategy}
    \label{fig:optimalAda}
    \end{subfigure}
    \begin{subfigure}[t]{0.22\textwidth}
    \centering
	\begin{tikzpicture}[scale=0.18]  
        \draw[->] (0,0) -- (0,10);
        \draw[->] (0,0) -- (10,0);

        \draw (10,0) node[below right] {$z$};
    	\draw (0,10) node[above left] {$y$};

        \fill[red, opacity=0.4] (3,3) -- (3,9) -- (9,9) -- (9,3) -- cycle;
        \fill[blue, opacity=0.4] (0,0) -- (9,0) -- (9,3) -- (3,3) -- (3,9) -- (0,9) -- cycle;
        \draw[red] (6,6) node {$\H$};
        \draw[blue!70!black] (1.5,1.5) node {$\T$};
    \end{tikzpicture}
    \caption{Bertrand's strategy}
    \label{fig:optimalBertrand}
    \end{subfigure}
    \begin{subfigure}[t]{0.22\textwidth}
    \centering
	\begin{tikzpicture}[scale=0.18]  
        \draw[->] (0,0) -- (0,10);
        \draw[->] (0,0) -- (10,0);

        \draw (10,0) node[below right] {$x$};
    	\draw (0,10) node[above left] {$z$};

        \fill[red, opacity=0.4] (3,3) -- (3,9) -- (9,9) -- (9,3) -- cycle;
        \fill[blue, opacity=0.4] (0,0) -- (9,0) -- (9,3) -- (3,3) -- (3,9) -- (0,9) -- cycle;
        \draw[red] (6,6) node {$\H$};
        \draw[blue!70!black] (1.5,1.5) node {$\T$};
    \end{tikzpicture}
    \caption{Claude's strategy}
    \label{fig:optimalClaude}
    \end{subfigure}
    \begin{subfigure}[t]{0.25\textwidth}
    \centering
	\begin{tikzpicture}[scale=0.2, line join=round] 
        \draw[->] (0,0,0) -- (10,0,0);
    	\draw[->] (0,0,0) -- (0,10,0);
    	\draw[->] (0,0,0) -- (0,0,10);

        \draw (10,0,0) node[below right] {$x$};
    	\draw (0,10,0) node[above left] {$y$};
    	\draw (0,0,10) node[left] {$z$};

        \fill[blue, opacity=0.4] (0,9,0) -- (3,9,0) -- (3,3,0) -- (9,3,0) -- (9,0,0) -- (0,0,0) -- cycle;
        \draw[thick] (0,9,0) -- (3,9,0) -- (3,3,0) -- (9,3,0) -- (9,0,0) -- (0,0,0) -- cycle; 

        \fill[blue, opacity=0.4] (0,0,9) -- (3,0,9) -- (3,0,3) -- (9,0,3) -- (9,0,0) -- (0,0,0) -- cycle;
        \draw[thick] (0,0,9) -- (3,0,9) -- (3,0,3) -- (9,0,3) -- (9,0,0) -- (0,0,0) -- cycle; 

        \fill[blue, opacity=0.4] (0,0,9) -- (0,3,9) -- (0,3,3) -- (0,9,3) -- (0,9,0) -- (0,0,0) -- cycle;
        \draw[thick] (0,0,9) -- (0,3,9) -- (0,3,3) -- (0,9,3) -- (0,9,0) -- (0,0,0) -- cycle; 
    
        \fill[blue, opacity=0.4] (0,9,0) -- (3,9,0) -- (3,9,3) -- (0,9,3) -- cycle;
        \draw[thick] (0,9,0) -- (3,9,0) -- (3,9,3) -- (0,9,3) -- cycle; 

        \fill[blue, opacity=0.4] (9,0,0) -- (9,3,0) -- (9,3,3) -- (9,0,3) -- cycle;
        \draw[thick] (9,0,0) -- (9,3,0) -- (9,3,3) -- (9,0,3) -- cycle; 

        \fill[blue, opacity=0.4] (0,0,9) -- (3,0,9) -- (3,3,9) -- (0,3,9) -- cycle;
        \draw[thick] (0,0,9) -- (3,0,9) -- (3,3,9) -- (0,3,9) -- cycle; 

        \fill[blue, opacity=0.4] (0,3,3) -- (3,3,3) -- (3,9,3) -- (0,9,3) -- cycle;
        \fill[blue, opacity=0.4] (0,3,3) -- (3,3,3) -- (3,3,9) -- (0,3,9) -- cycle;
        \fill[blue, opacity=0.4] (3,0,3) -- (3,3,3) -- (3,3,9) -- (3,0,9) -- cycle;
        \fill[blue, opacity=0.4] (3,3,0) -- (3,3,3) -- (3,9,3) -- (3,9,0) -- cycle;
        \fill[blue, opacity=0.4] (3,3,0) -- (3,3,3) -- (9,3,3) -- (9,3,0) -- cycle;
        \fill[blue, opacity=0.4] (3,0,3) -- (3,3,3) -- (9,3,3) -- (9,0,3) -- cycle;

        \newcommand{\rcuboid}[6]{
          \pgfmathsetmacro{\x}{#1}
          \pgfmathsetmacro{\y}{#2}
          \pgfmathsetmacro{\z}{#3}
          \pgfmathsetmacro{\xa}{\x+#4}
          \pgfmathsetmacro{\ya}{\y+#5}
          \pgfmathsetmacro{\za}{\z+#6}
        
          \fill[red, opacity=0.4] (\x,\y,\z) -- (\xa,\y,\z) -- (\xa,\ya,\z) -- (\x,\ya,\z) -- cycle; 
          \fill[red, opacity=0.4] (\x,\y,\z) -- (\xa,\y,\z) -- (\xa,\y,\za) -- (\x,\y,\za) -- cycle; 
          \fill[red, opacity=0.4] (\x,\y,\z) -- (\x,\ya,\z) -- (\x,\ya,\za) -- (\x,\y,\za) -- cycle; 
          \fill[red, opacity=0.4] (\xa,\ya,\za) -- (\x,\ya,\za) -- (\x,\y,\za) -- (\xa,\y,\za) -- cycle; 
          \fill[red, opacity=0.4] (\xa,\ya,\za) -- (\xa,\y,\za) -- (\xa,\y,\z) -- (\xa,\ya,\z) -- cycle; 
          \fill[red, opacity=0.4] (\xa,\ya,\za) -- (\xa,\ya,\z) -- (\x,\ya,\z) -- (\x,\ya,\za) -- cycle; 
        
          \draw[thick] (\x,\y,\z) -- (\xa,\y,\z) -- (\xa,\ya,\z) -- (\x,\ya,\z) -- cycle; 
          \draw[thick] (\x,\y,\za) -- (\xa,\y,\za) -- (\xa,\ya,\za) -- (\x,\ya,\za) -- cycle; 
          \draw[thick] (\x,\y,\z) -- (\x,\y,\za);
          \draw[thick] (\xa,\y,\z) -- (\xa,\y,\za);
          \draw[thick] (\xa,\ya,\z) -- (\xa,\ya,\za);
          \draw[thick] (\x,\ya,\z) -- (\x,\ya,\za);
        }

        \rcuboid{3}{3}{3}{6}{6}{6}
    \end{tikzpicture}
    \caption{The sets $\H\H\H$ and $\T\T\T$}
    \label{fig:optimalHHH}
    \end{subfigure}
	\caption{An optimal collective strategy} \label{fig:optimalstrategy}
\end{figure}

\subsection{Motivation}
In the context of zero-sum games, dicey games represent a natural intermediate notion between the independent randomisation studied by Cournot~\cite{Cournot1838}, von Neumann and Morgenstern~\cite{vNM44}, and Nash~\cite{Nas50}, and the shared randomness inherent to correlated equilibria~\cite{Aum1974,AD74}. 
In correlated equilibria, all players can coordinate using one common source of randomness, which makes it a special case of ours.
Universal access to shared randomness is rarely realistic. As Goldreich notes in his book on pseudorandomness~\cite{Gol10}:
    ``randomness is not an inherent property of an object, but is rather subjective to the observer.'' 
This motivates 
modelling the topology of shared randomness as a part of the input, where correlation is limited by the specific sources an agent can observe.

Similar to how such concurrent games have found its uses in economics, formal verification~\cite{Umm10,kwiatkowska2018equilibria,AGHHKNPSW21}, economics~\cite{EconNash99}, social choice~\cite{HOWARD1992142,socialchoiceself2002} and even biology~\cite{smith1982evolution}, we foresee that our extension also models several settings 
in these fields. We highlight some areas and results which support this claim. 

\subparagraph*{Computer Science.}
Unrestricted access to shared randomness can degrade privacy (for example, in systems that are involved in secret sharing)~\cite{Bellare2020}, it may bias system behaviour~\cite{Syta2017}, create unintended communication channels~\cite{Simmons1984}, or violate correctness assumptions of distributed protocols~\cite{Gilad2017}. 
By modelling interactive and distributed systems that need to be verified under restricted randomness as games, we can enforce independence where needed, prevent accidental information leakage, maintain fairness and privacy guarantees, and model realistic constraints. 
Concurrent game models for computer science have often used the assumption that players belonging to a same coalition can resort to a shared source of randomness, but some recent work~\cite{DBLP:journals/corr/abs-2601-18303} has also considered the case of individual randomisation for games played on graphs.

\subparagraph*{Economics.}
Agents’ strategies in financial settings are more accurately modelled by different topologies of shared randomness~\cite{bergemann2016bayes}. 
Firms, consumers, and institutions often face legal, informational, or organizational constraints that limit their ability to randomise~\cite{kuhn2001fighting, vives2011information}. Such restrictions prevent tacit coordination and collusion, preserve decentralised decision-making, and sustain incentive compatibility~\cite{hayek1945use}. 

\subparagraph*{Social Choice.}
Shared randomness plays a central role in resolving ties~\cite{brandt2019collective}, aggregating preferences~\cite{brandl2016consistent}, and ensuring fairness~\cite{budish2013designing}, yet unrestricted access to common random signals can alter collective outcomes~\cite{Aum74b}. 
Correlation structure among agents is not merely a technical detail but a determinant of welfare; for instance, shared information sources can coordinate voters or agents onto specific collective outcomes that are otherwise unstable or suboptimal~\cite{alonso2016persuading,bergemann2019information}. %
Conversely, restricting shared randomness, such as in cases where only certain groups or institutions observe particular random signals, reflects realistic features of voting systems, committees, and media-consumption. 


\subsection{Structure of the paper}
We introduce the necessary definitions in \Cref{sec:definitions}.
In \cref{sec:example}, we start with a pictorial treatment of the triangular matching-pennies game, 
and then generalise our reasoning to prove the optimality of $k$-grid strategies.
In \cref{sec:optimal}, we analyse the descriptive complexity of optimal collective strategies, showing that both the optimal $\max\min$ and the strategies achieving it require at most exponentially many bits to be represented.
In \Cref{sec:algos}, we prove several complexity results.
In \Cref{sec:moreexamples}, we conclude with some examples and  conjectures on matching pennies games.

As is customary, the more patient proofs are deferred to the appendix. We also collect in Appendix~\ref{app:hammers} a small toolbox of results from algebraic geometry and nonlinear optimisation. 
%

\section{Definitions} \label{sec:definitions}
Before rolling any dice,
we pause to fix notation and define the formal setup of the games under consideration.
Given an integer $n$, we write $[n]$ to denote the set $\{1, \dots, n\}$.

For a set $X$, and an indexing set $I$ (assumed to be clear from the context), when we have defined an element $x_i \in X$ for every $i \in I$, we usually use the bar notation $\bx$ for the tuple $(x_i)_{i \in I} \in X^I$.
Conversely, when we introduce a tuple $\bx$, the notation $x_i$ refers to the element of index $i \in I$.
An incomplete tuple $(x_i)_{i \in I \setminus \{j\}}$ is written $\bx_{-j}$, and the notation $(\bx_{-j}, y)$ denotes the tuple $\bz$ where $z_i = x_i$ for each $i \in I \setminus \{j\}$, and $z_j = y$.
Unfortunately, we sometimes have to use the double-bar $\bbx$ to denote tuples of tuples.

\subsection{Probabilities}

Given a set of outcomes $\Universe$ and a probability measure $\prob$ over $\Universe$, let $X$ be a random variable over $\Universe$, i.e., a mapping $X: \Universe \to \Rb$.
We write $\Eb^\prob[X]$, or simply $\Eb[X]$, for the expected value of $X$, when defined.
Given a finite set $S$, a \emph{probability distribution} over $S$ is a mapping $d: S \to [0,1]$ satisfying the equality $\sum_{x \in S} d(x) = 1$.

\subsection{Games}
Throughout this work, the word \emph{game} refers to matrix games played by a team of players against one opponent.

\begin{definition}[Game]
    A \emph{game} is a tuple $\Game = (\Players, \opp, (\Actions_\player)_{\player \in \Players \cup \{\opp\}}, \payoff)$,
    that consists of:
    \begin{itemize}
        \item a finite set $\Players \cup \{\opp\}$ of \emph{players}, where the players $\player \in \Players$ are called \emph{team players} and the player $\opp \not\in \Players$ is called \emph{Devil};

        \item for each player $\player \in \Players \cup \{\opp\}$, a set $\Actions_\player$ of \emph{actions};

        \item a mapping $\payoff: \prod_{\player \in \Players \cup \{\opp\}} \Actions_\player \to \Zb$ called \emph{payoff function}, that indicates the reward that the team gets when a given tuple of actions is selected.
    \end{itemize}
\end{definition}

A game where all payoffs are either $0$ or $1$ is a \emph{Boolean} game.
Then, the team \emph{wins} when it gets payoff $1$, and \emph{loses} otherwise.

\subsection{Dice}

Along with the game structure, we assume that the players are given access to some \emph{dice}, that is, to external sources of randomness, which may be shared by several players.

\begin{definition}[Dice structure, dicey game]
    A \emph{dice pack} is a pair $\Pack = (\Dice, \acc)$, where $\Dice$ is a set of \emph{dice}, i.e., of independent random variables following the uniform distribution over $[0,1]$, and where the mapping $\acc: \Dice \to \Nb$ maps each die to its \emph{accessibility}---intuitively, the maximal number of players that can access it.

    A \emph{dice structure}, for a given team $\Players$, is a pair $\Struct = (\Dice, \Acc)$, where $\Dice$ is a set of dice, and the mapping $\Acc: \Dice \to 2^{\Players}$ maps each die to the set of players who have \emph{access} to it.
    The dice structure $\Struct$ \emph{matches} the dice pack $\Pack$ if for each die $\Die \in \Dice$, we have $|\Acc(\Die)| \leq \acc(\Die)$.

    A \emph{dicey game} is a pair $\DGame = (\Game, \Struct)$, where $\Game$ is a game, and $\Struct$ is a dice structure for the team players in $\Game$.
\end{definition}

When referring to a game $\Game$, a dice pack $\Pack$, a dice structure $\Struct$, or a dicey game $\DGame$, we will often use the notations $\Pi, \mu, \Delta, \Acc$, etc., given above without necessarily recalling them.
We then write $\Dice_\player = \{\Die \in \Dice \mid \player \in \Acc(\Die)\}$ for the set of dice that player $\player$ can access.
Note that in our definition, the Devil doesn't have access to any die: this is because such a die would be useless for the team players, since the purpose of randomisation is to hide information from him\footnote{For a more rigorous argument: assume that the Devil has access to a given die $\Die$. Consider a collective strategy $\bstrat$ (see \Cref{def:strategy} for the formal definition). Then, consider a value $x$ such that the team's expected payoff, following $\bstrat$, assuming that the Devil responds optimally, and knowing that the die $\Die$ equals $x$, is maximised. Then, the collective strategy that ignores $\Die$ and play as if its outcome was always $x$ is at least as good as $\bstrat$.}.
For convenience, we define $\Dice_\opp$ as equal to the empty set.

We call \emph{roll of the die $\Die$} an outcome of the random variable $\Die$.
A \emph{roll} (with no mention of a specific die) is a tuple $\broll \in [0, 1]^\Dice$, and a \emph{roll for player $\player$} is a tuple $\broll \in [0, 1]^{\Dice_\player}$.

\subsection{Strategies, and strategy profiles}

Strategies are defined so that the players are not allowed to randomise their actions by themselves, but can do so by using the dice they have access to.

\begin{definition}[Strategy, collective strategy, strategy profile]\label{def:strategy}
    A \emph{strategy} for player $\everyplayer$ in the dicey game $\DGame$ is a mapping $\strat_\player: [0, 1]^{\Dice_\player} \to \Actions_\player$.
    A \emph{collective strategy} is a tuple $\bstrat = (\strat_\player)_{\player \in \Players}$, where each $\strat_\player$ is a strategy for player $\player$.
    A \emph{strategy profile} is a tuple $(\bstrat, \strat_\opp) = (\strat_\player)_{\everyplayer}$.
\end{definition}

Players that have access to no die can play only deterministic strategies.
In particular, this is always the case for the Devil.
Since the Devil is always assumed to know the team's strategy, and to define his strategy accordingly, this is not a restriction: he can choose the action that minimises the team's expected payoff.
For convenience, we sometimes also use an action $\action \in \Actions_\dev$ to represent the strategy that deterministically picks the action $\action$.
We also sometimes write $\bstrat(\broll)$, for $\roll \in [0, 1]^{\Delta}$, for the tuple of actions that is played when the roll is $\broll$, under the collective strategy~$\bstrat$.

A strategy profile $(\bstrat, \strat_\opp)$ defines a probability distribution over the set $\prod_\everyplayer \Actions_\player$, that we write as $\prob_{\bstrat, \strat_\opp}$.
We write $\Eb(\bstrat, \strat_\opp)$ for the expected payoff $\Eb^{\prob_{\bstrat, \strat_\opp}}[\payoff]$.
We define the \emph{value} of a collective strategy $\bstrat$ as the quantity
$\val(\bstrat) = \min_{\action \in \Actions_\opp} \Eb(\bstrat, \action).$
A collective strategy is \emph{optimal} if no strategy has a larger value.


\subsection{Grid strategies}

A strategy for player $\player$ can be seen as a partition of the roll space $[0,1]^{\Dice_\player}$ into regions, each of which corresponds to an action available to player $\player$.
In the next sections, we will show that strategies such that this partition is defined with straight lines are sufficient.

\begin{definition}[Piecewise constant strategy, grid collective strategy]
    Let $\bstrat$ be a collective strategy, and let $\Die$ be a die.
    We say that $\bstrat$ is \emph{$k$-piecewise constant for the die $\Die$} if, for every roll $\broll_{-\Die} \in [0, 1]^{\Dice \setminus \{\Die\}}$, the function $\roll_\Die \mapsto \bstrat(\broll_{-\Die}, \roll_\Die)$ is piecewise constant with at most $k$ pieces.
    The collective strategy $\bstrat$ is \emph{$k$-grid} if it is $k$-piecewise constant for each die $\Die$.
\end{definition}

For examples, the reader may refer to the last dicey game described in the introduction, and to the strategies that have been presented for this game.
The collective strategy depicted by \Cref{fig:tetrahedron}, for instance, is not $k$-grid for any $k$.
The one illustrated by \Cref{fig:optimalstrategy}, for comparison, is $2$-grid.
In \Cref{fig:sigma0}, finally, the reader will find a collective strategy that is $3$-grid, but not $2$-grid.

\begin{figure}[h]
    \centering
    \begin{subfigure}[t]{0.25\textwidth}
    \centering
	\begin{tikzpicture}[scale=0.18]  
        \draw[->] (0,0) -- (0,10);
        \draw[->] (0,0) -- (10,0);

        \draw (10,0) node[below right] {$x$};
    	\draw (0,10) node[above left] {$y$};

        \fill[red, opacity=0.4] (0,0) -- (3,0) -- (3,3) -- (0,3) -- cycle;
        \fill[red, opacity=0.4] (6,0) -- (9,0) -- (9,3) -- (6,3) -- cycle;
        \fill[red, opacity=0.4] (0,6) -- (3,6) -- (3,9) -- (0,9) -- cycle;
        \fill[red, opacity=0.4] (6,6) -- (9,6) -- (9,9) -- (6,9) -- cycle;
        \fill[blue, opacity=0.4] (3,3) -- (6,3) -- (6,6) -- (3,6) -- cycle;
        \fill[blue, opacity=0.4] (0,3) -- (3,3) -- (3,6) -- (0,6) -- cycle;
        \fill[blue, opacity=0.4] (6,3) -- (9,3) -- (9,6) -- (6,6) -- cycle;
        \fill[blue, opacity=0.4] (3,6) -- (6,6) -- (6,9) -- (3,9) -- cycle;
        \fill[blue, opacity=0.4] (3,0) -- (6,0) -- (6,3) -- (3,3) -- cycle;

        \draw[red] (1.5,1.5) node {$\H$};
        \draw[red] (7.5,7.5) node {$\H$};
        \draw[red] (7.5,1.5) node {$\H$};
        \draw[red] (1.5,7.5) node {$\H$};
        \draw[blue!70!black] (4.5,4.5) node {$\T$};
    \end{tikzpicture}
    \caption{Ada's strategy}
    \label{fig:sigma0Ada}
    \end{subfigure}
    \begin{subfigure}[t]{0.25\textwidth}
    \centering
	\begin{tikzpicture}[scale=0.18]  
        \draw[->] (0,0) -- (0,10);
        \draw[->] (0,0) -- (10,0);

        \draw (10,0) node[below right] {$z$};
    	\draw (0,10) node[above left] {$y$};

        \fill[red, opacity=0.4] (0,0) -- (3,0) -- (3,3) -- (0,3) -- cycle;
        \fill[red, opacity=0.4] (6,0) -- (9,0) -- (9,3) -- (6,3) -- cycle;
        \fill[red, opacity=0.4] (0,6) -- (3,6) -- (3,9) -- (0,9) -- cycle;
        \fill[red, opacity=0.4] (6,6) -- (9,6) -- (9,9) -- (6,9) -- cycle;
        \fill[blue, opacity=0.4] (3,3) -- (6,3) -- (6,6) -- (3,6) -- cycle;
        \fill[blue, opacity=0.4] (0,3) -- (3,3) -- (3,6) -- (0,6) -- cycle;
        \fill[blue, opacity=0.4] (6,3) -- (9,3) -- (9,6) -- (6,6) -- cycle;
        \fill[blue, opacity=0.4] (3,6) -- (6,6) -- (6,9) -- (3,9) -- cycle;
        \fill[blue, opacity=0.4] (3,0) -- (6,0) -- (6,3) -- (3,3) -- cycle;

        \draw[red] (1.5,1.5) node {$\H$};
        \draw[red] (7.5,7.5) node {$\H$};
        \draw[red] (7.5,1.5) node {$\H$};
        \draw[red] (1.5,7.5) node {$\H$};
        \draw[blue!70!black] (4.5,4.5) node {$\T$};
    \end{tikzpicture}
    \caption{Bertrand's strategy}
    \label{fig:sigma0Bertrand}
    \end{subfigure}
    \begin{subfigure}[t]{0.25\textwidth}
    \centering
	\begin{tikzpicture}[scale=0.18]  
        \draw[->] (0,0) -- (0,10);
        \draw[->] (0,0) -- (10,0);

        \draw (10,0) node[below right] {$x$};
    	\draw (0,10) node[above left] {$z$};

        \fill[red, opacity=0.4] (0,0) -- (3,0) -- (3,3) -- (0,3) -- cycle;
        \fill[red, opacity=0.4] (6,0) -- (9,0) -- (9,3) -- (6,3) -- cycle;
        \fill[red, opacity=0.4] (0,6) -- (3,6) -- (3,9) -- (0,9) -- cycle;
        \fill[red, opacity=0.4] (6,6) -- (9,6) -- (9,9) -- (6,9) -- cycle;
        \fill[blue, opacity=0.4] (3,3) -- (6,3) -- (6,6) -- (3,6) -- cycle;
        \fill[blue, opacity=0.4] (0,3) -- (3,3) -- (3,6) -- (0,6) -- cycle;
        \fill[blue, opacity=0.4] (6,3) -- (9,3) -- (9,6) -- (6,6) -- cycle;
        \fill[blue, opacity=0.4] (3,6) -- (6,6) -- (6,9) -- (3,9) -- cycle;
        \fill[blue, opacity=0.4] (3,0) -- (6,0) -- (6,3) -- (3,3) -- cycle;

        \draw[red] (1.5,1.5) node {$\H$};
        \draw[red] (7.5,7.5) node {$\H$};
        \draw[red] (7.5,1.5) node {$\H$};
        \draw[red] (1.5,7.5) node {$\H$};
        \draw[blue!70!black] (4.5,4.5) node {$\T$};
    \end{tikzpicture}
    \caption{Claude's strategy}
    \label{fig:sigma0Claude}
    \end{subfigure}
    \begin{subfigure}[t]{0.35\textwidth}
    \centering
	\begin{tikzpicture}[scale=0.25, line join=round]  
        
        \def\s{3} 
        
        \newcommand{\rcube}[3]{%
          \pgfmathsetmacro{\x}{#1}
          \pgfmathsetmacro{\y}{#2}
          \pgfmathsetmacro{\z}{#3}
          \pgfmathsetmacro{\xa}{\x+\s}
          \pgfmathsetmacro{\ya}{\y+\s}
          \pgfmathsetmacro{\za}{\z+\s}
        
          \fill[red, opacity=0.4] (\x,\y,\z) -- (\xa,\y,\z) -- (\xa,\ya,\z) -- (\x,\ya,\z) -- cycle; 
          \fill[red, opacity=0.4] (\x,\y,\z) -- (\xa,\y,\z) -- (\xa,\y,\za) -- (\x,\y,\za) -- cycle; 
          \fill[red, opacity=0.4] (\x,\y,\z) -- (\x,\ya,\z) -- (\x,\ya,\za) -- (\x,\y,\za) -- cycle; 
          \fill[red, opacity=0.4] (\xa,\ya,\za) -- (\x,\ya,\za) -- (\x,\y,\za) -- (\xa,\y,\za) -- cycle; 
          \fill[red, opacity=0.4] (\xa,\ya,\za) -- (\xa,\y,\za) -- (\xa,\y,\z) -- (\xa,\ya,\z) -- cycle; 
          \fill[red, opacity=0.4] (\xa,\ya,\za) -- (\xa,\ya,\z) -- (\x,\ya,\z) -- (\x,\ya,\za) -- cycle; 
        
          \draw[thick] (\x,\y,\z) -- (\xa,\y,\z) -- (\xa,\ya,\z) -- (\x,\ya,\z) -- cycle; 
          \draw[thick] (\x,\y,\za) -- (\xa,\y,\za) -- (\xa,\ya,\za) -- (\x,\ya,\za) -- cycle; 
          \draw[thick] (\x,\y,\z) -- (\x,\y,\za);
          \draw[thick] (\xa,\y,\z) -- (\xa,\y,\za);
          \draw[thick] (\xa,\ya,\z) -- (\xa,\ya,\za);
          \draw[thick] (\x,\ya,\z) -- (\x,\ya,\za);
        }
        
        \newcommand{\gcube}[3]{%
          \pgfmathsetmacro{\x}{#1}
          \pgfmathsetmacro{\y}{#2}
          \pgfmathsetmacro{\z}{#3}
          \pgfmathsetmacro{\xa}{\x+\s}
          \pgfmathsetmacro{\ya}{\y+\s}
          \pgfmathsetmacro{\za}{\z+\s}
        
          \fill[blue, opacity=0.4] (\x,\y,\z) -- (\xa,\y,\z) -- (\xa,\ya,\z) -- (\x,\ya,\z) -- cycle; 
          \fill[blue, opacity=0.4] (\x,\y,\z) -- (\xa,\y,\z) -- (\xa,\y,\za) -- (\x,\y,\za) -- cycle; 
          \fill[blue, opacity=0.4] (\x,\y,\z) -- (\x,\ya,\z) -- (\x,\ya,\za) -- (\x,\y,\za) -- cycle; 
          \fill[blue, opacity=0.4] (\xa,\ya,\za) -- (\x,\ya,\za) -- (\x,\y,\za) -- (\xa,\y,\za) -- cycle; 
          \fill[blue, opacity=0.4] (\xa,\ya,\za) -- (\xa,\y,\za) -- (\xa,\y,\z) -- (\xa,\ya,\z) -- cycle; 
          \fill[blue, opacity=0.4] (\xa,\ya,\za) -- (\xa,\ya,\z) -- (\x,\ya,\z) -- (\x,\ya,\za) -- cycle; 
        
          \draw[thick] (\x,\y,\z) -- (\xa,\y,\z) -- (\xa,\ya,\z) -- (\x,\ya,\z) -- cycle; 
          \draw[thick] (\x,\y,\za) -- (\xa,\y,\za) -- (\xa,\ya,\za) -- (\x,\ya,\za) -- cycle; 
          \draw[thick] (\x,\y,\z) -- (\x,\y,\za);
          \draw[thick] (\xa,\y,\z) -- (\xa,\y,\za);
          \draw[thick] (\xa,\ya,\z) -- (\xa,\ya,\za);
          \draw[thick] (\x,\ya,\z) -- (\x,\ya,\za);
        }

        \draw[->] (0,0,0) -- (10,0,0);
    	\draw[->] (0,0,0) -- (0,10,0);
    	\draw[->] (0,0,0) -- (0,0,10);
    	
    	\draw (10,0,0) node[below right] {$x$};
    	\draw (0,10,0) node[above left] {$y$};
    	\draw (0,0,10) node[left] {$z$};

        \fill[black, opacity=0.3] (1,0,0)--(1,9,0)--(1,9,9)--(1,0,9);

        \rcube{0}{0}{0}
        \rcube{0}{0}{6}
        \rcube{0}{6}{0}
        \rcube{0}{6}{6}

        \draw[very thick, black, dashed] (1,0,0)--(1,3,0)--(1,3,3)--(1,0,3)--cycle;
        \draw[very thick, black, dashed] (1,6,0)--(1,9,0)--(1,9,3)--(1,6,3)--cycle;
        \draw[very thick, black, dashed] (1,0,6)--(1,3,6)--(1,3,9)--(1,0,9)--cycle;
        \draw[very thick, black, dashed] (1,6,6)--(1,9,6)--(1,9,9)--(1,6,9)--cycle;

        \fill[black, opacity=0.3] (5,0,0)--(5,9,0)--(5,9,6)--(5,0,6);

        \gcube{3}{3}{3}

        \fill[black, opacity=0.3] (5,0,6)--(5,9,6)--(5,9,9)--(5,0,9);

        \draw[very thick, black, dashed] (5,3,3)--(5,3,6)--(5,6,6)--(5,6,3)--cycle;

        \rcube{6}{0}{0}
        \rcube{6}{0}{6}
        \rcube{6}{6}{0}
        \rcube{6}{6}{6}
    \end{tikzpicture}
    \caption{The sets $\H\H\H$ and $\T\T\T$}
    \label{fig:sigma0HHH}
    \end{subfigure}
    \begin{subfigure}[t]{0.2\linewidth}
    \centering
	\begin{tikzpicture}[scale=0.1]  
        \draw[->] (0,0) -- (0,10);
        \draw[->] (0,0) -- (10,0);

        \draw (10,0) node[below right] {$z$};
    	\draw (0,10) node[above left] {$y$};

        \draw[very thick, black] (0,0)--(0,9)--(9,9)--(9,0)--cycle;

        \fill[red, opacity=0.4] (0,0) -- (3,0) -- (3,3) -- (0,3) -- cycle;
        \fill[red, opacity=0.4] (6,0) -- (9,0) -- (9,3) -- (6,3) -- cycle;
        \fill[red, opacity=0.4] (0,6) -- (3,6) -- (3,9) -- (0,9) -- cycle;
        \fill[red, opacity=0.4] (6,6) -- (9,6) -- (9,9) -- (6,9) -- cycle;
    \end{tikzpicture}
    \caption{The slice $x=0.1$}
    \label{fig:sigma0x=1}
    \end{subfigure}
    \begin{subfigure}[t]{0.2\linewidth}
    \centering
	\begin{tikzpicture}[scale=0.1]  
        \draw[->] (0,0) -- (0,10);
        \draw[->] (0,0) -- (10,0);

        \draw (10,0) node[below right] {$z$};
    	\draw (0,10) node[above left] {$y$};

        \draw[very thick, black] (0,0)--(0,9)--(9,9)--(9,0)--cycle;

        \fill[blue, opacity=0.4] (3,3) -- (6,3) -- (6,6) -- (3,6) -- cycle;
    \end{tikzpicture}
    \caption{The slice $x=0.5$}
    \label{fig:sigma0x=5}
    \end{subfigure}
	\caption{The collective strategy $\bstrat^0$} \label{fig:sigma0}
\end{figure}

\subsection{Size}

We use the word \emph{size} to refer to the \emph{bit-size} of the object we manipulate, i.e., the number of bits that is required to describe them, in a canonical encoding.
In particular, the size of an integer $n$ is $\log_2\lceil n + 1 \rceil$, and the size of a rational number $p/q$ is the size of $p$, plus the size of $q$, plus $1$.
The size of a composite object, such as a tuple, a finite set, a polynomial, etc., is obtained by summing the sizes of the objects it is composed of, and adding $1$.
The size of an algebraic number is the size of its minimal polynomial.

As for games, in the general case, the space required to describe the mapping $\payoff$ is exponential in the number of players.
However, in practice, many games can be described succinctly.
For complexity considerations, we will assume that this mapping is described by a list $(\baction^1_{P_1}, x_1), \dots, (\baction^n_{P_n}, x_n), x_{n+1}$, where $P_k \subseteq \Players \cup \{\opp\}$ and $\baction^k_{P_k} \in \prod_{\player \in P_k} A_\player$ for each $k$, serving as an if-then-else list: for a given action tuple $\baction$, 
\begin{itemize}
    \item if we have $\action_\player = \baction^1_\player$ for each $\player \in P_1$, then we have $\payoff(\baction) = x_1$;

    \item else, if we have $\action_\player = \baction^2_\player$ for each $\player \in P_2$, then we have $\payoff(\baction) = x_2$;

    \item[] \dots

    \item else, we have $\payoff(\baction) = x_{n+1}$.
\end{itemize}
For example, in a matching pennies game in which all players choose between the actions Heads (written $\H$) and Tails (written $\T$), the fact that the team gets payoff $1$ if all actions match, and $0$ otherwise, can be described by the list $\left((\H)_\everyplayer, 1\right), \left((\T)_\everyplayer, 1\right), 0$, which takes only space $O(|\Players|)$.

\section{Dice and slice to make strategies nice} \label{sec:example}

The main result of this section, \cref{thm:itsOKtobestraight}, states that whenever there exists a collective  strategy achieving value $t$, there also exists a grid strategy achieving it.

Our approach is constructive. Given an arbitrary collective strategy, we show how to transform it into a grid strategy with the same value. Each step of the transformation preserves value while progressively simplifying the geometric structure.

To build intuition, we first consider the example that was described in the introduction: triangular matching pennies.
We prove our theorem in this specific example, and use it to show that the $\max\min$ value of this game is approximately $0.2781$.
In a second subsection, we generalise the reasoning to all dicey games.

The key technical step is a geometric reshaping of collective strategies via \emph{slicing}. Recall that each player’s strategy can be viewed as a partition of the unit square. Fixing the outcome of one of the two dice restricts the strategy to a lower-dimensional object, which we refer to as a \emph{slice}. Visually, this corresponds to taking a cross-section of the strategy along one coordinate. We show that any given collective strategy can be modified, without reducing its value, into one that is constructed from finitely many (two in the case of our example) slices. Repeating this argument across all dice yields a grid strategy ($2$-grid in our example). 
The justification that finitely many slices suffice relies on Carathéodory’s theorem.


\subsection{Triangular matching pennies} \label{ssec:triangular}

\begin{definition}[The dicey game $\DGame^\triangledown$]
    The three team players are Ada, Bertrand, and Claude.
    All players, including the Devil, choose simultaneously either the action \emph{Heads ($\H$)} or \emph{Tails ($\T$)}.
    The team wins if all choices coincide, and loses otherwise.
    The team players can use three dice, called $D_1$, $D_2$, and $D_3$: Ada has access to $D_1$ and $D_2$, Bertrand to $D_2$ and $D_3$, and Claude to $D_3$ and $D_1$.
\end{definition}

In what follows, we write $\alpha$ for the unique root of the polynomial $X^3 - 3X + 1=0$ that lies in the interval $[0, 1]$ (approximately $\alpha \approx 0.3473$), and we write $\beta = 3 \alpha^2 - 2\alpha^3 \approx 0.2781$.
This subsection is dedicated to proving the following theorem.

\begin{theorem} \label{thm:3v1example}
    The team has a collective strategy optimal strategy with value $\beta\approx 0.2781$.
\end{theorem}

Let us first show a collective strategy that realises that value.

\begin{lemma}
    The team has a collective strategy of value at least $\beta$.
\end{lemma}

\begin{proof}
    Let us define the collective strategy $\bstrat$ as follows: each team player picks the action Heads if the two rolls they have access to are greater than $\alpha$, and Tails otherwise.
    This collective strategy is illustrated by \cref{fig:optimalstrategy}, where $x$, $y$, and $z$ are the rolls of the dice $\Die_1$, $\Die_2$, and $\Die_3$, respectively.
    If the Devil responds to this strategy by always playing    Tails, then the team wins if all three team players choose Tails, which happens with probability $3\alpha^2 - 2\alpha^3 = \beta$.
    If he plays Heads, the team wins with probability $(1-\alpha)^3 = \beta$.
    The value of this collective strategy is thus exactly $\beta$.
\end{proof}

To show that no better value can be achieved, we first prove that $2$-grid strategies are optimal in this game; later, we will prove that the collective strategy presented in the previous proof is optimal among $2$-grid strategies.

\begin{lemma} \label{lm:itsOKtobestraightexample}
    Let $\bstrat^0$ be a collective strategy.
    Then, there exists a $2$-grid collective strategy $\bstrat^\star$ with value at least as high as $\bstrat^0$.
\end{lemma}

\begin{proof}
    In this proof, we denote a roll by $(x, y, z)$, where $x, y, z$ are the outcomes of the dice $D_1$, $D_2$, and $D_3$, respectively.

    Our result will be proven by showing that from any collective strategy and for any die, we can construct a new collective strategy, which is at least as performant and is $2$-piecewise constant for the specified die.
    That is done by choosing carefully two "slices", and constructing new strategies only from those.
    Applying that transformation for all dice, we obtain a $2$-grid collective strategy.

    This process is illustrated by \Cref{fig:sigma0,fig:sigma1,fig:sigma2,fig:sigmastar}.
    If $\bstrat^0$ is the original collective strategy, then the collective strategy $\bstrat^1$ is obtained by applying that transformation for the die $D_1$; then, the collective strategy $\bstrat^2$ is obtained by applying it for the die $D_2$; and finally, the collective strategy $\bstrat^\star$ is obtained by applying it for the dice $D_3$.
    For each of them, we have also depicted, in the cube $[0,1]^3$, the parts of the space in which the choices of Ada, Bertrand, and Claude align: red when all pick Heads, blue when all pick Tails.

    \paragraph*{Definition of the transformation}

    Let $\bstrat$ be a collective strategy. 
    For $x \in [0, 1]$, we define $f(x)$ and $g(x)$ as the probabilities of all team players choosing $\H$ or all choosing $\T$, respectively, knowing that the outcome of the die $D_1$ is $x$, and that the team players are following the collective strategy $\bstrat$:
    $$f: x \mapsto \prob_{\bstrat}\left(\left. \strat_\ada = \strat_\bert = \strat_\clau = \H ~\right|~ D_1 = x \right)$$
    $$g: x \mapsto \prob_{\bstrat}\left(\left. \strat_\ada = \strat_\bert = \strat_\clau = \T ~\right|~ D_1 = x \right).$$

    Geometrically, if $\H\H\H$ denotes the part of the cube of rolls $[0,1]^3$ in which the three players choose $\H$, then $f(x)$ is the area of the slice $\H\H\H \cap (\{x\} \times [0,1]^2)$, and similarly for $g$.

    Let us note that the probability of all players choosing $\H$ is then $\int f$, and that the probability of all choosing $\T$ is $\int g$ (without more precision, integrals are considered on the whole domain of the integrand, that is, the interval $[0, 1]$).

    Now, let $V = \{(f(x), g(x)) \mid x \in [0, 1]\}$ be the set of values taken by the pair $(f, g)$.
    The geometrical arguments of the following paragraph are illustrated by \cref{fig:VTQ}.
    The pair $\left(\int f, \int g\right)$ lies in the convex hull of $V$, which we denote by $\Conv (V)$.
    By Carathéodory's theorem, since the set $V$ is contained in a space of dimension $2$, that point can therefore be written as a convex combination of only $3$ points of $V$: there exist $(u_1, v_1), (u_2, v_2), (u_3, v_3) \in V$ such that:
    $$\left(\int f, \int g\right) \in \Conv\left\{ (u_1, v_1), (u_2, v_2), (u_3, v_3)\right\}.$$
    Now, let $T$ be the triangle $\Conv\left\{ (u_1, v_1), (u_2, v_2), (u_3, v_3)\right\}$ and $Q$ be the quarter of plane $Q = \left[ \int f, +\infty\right) \times \left[ \int g, +\infty\right)$.
    We know that the triangle $T$ intersects $Q$, since the intersection contains at least the pair $(\int f, \int g)$.
    Then, at least one edge of $T$, which we write $\Conv\{(u_i, v_i), (u_j, v_j)\}$, intersects $Q$.
    Let us pick one point $(u^\star, v^\star) \in \Conv\{(u_i, v_i), (u_j, v_j)\} \cap Q$, and let $\lambda \in [0, 1]$ be such that:
    $$(u^\star, v^\star) = \lambda (u_i, v_i) + (1-\lambda) (u_j, v_j).$$

\begin{figure} 
    \centering
	\begin{tikzpicture}[scale=0.3]
    	\draw[->] (0,0) -- (10,0);
    	\draw[->] (0,0) -- (0,10);
    	
    	\draw (10,0) node[right] {$u$};
    	\draw (0,10) node[above] {$v$};

    	\fill[blue!40, opacity=0.2]
      plot [smooth cycle, tension=0.8]
      coordinates {(1,8) (1.5,9.2) (4,9.5) (3.5,8.2) (2.8,7) (1.2,7)};
        \fill[blue!40, opacity=0.2]
      plot [smooth cycle, tension=0.8]
      coordinates {(0.5,1) (0.4,3) (1,4) (3, 3) (4, 2) (2,2)};
        \fill[blue!40, opacity=0.2]
      plot [smooth cycle, tension=0.8]
      coordinates {(7, 0.5) (9.5, 1) (9, 4) (6, 1.5)};
      
    \fill[violet!5] (5, 5)--(5, 10)--(10, 10)--(10,5);

        \fill[orange!60, opacity=0.1] (3, 9)--(1, 2)--(9, 3)--(3, 9);
        \draw[orange] (3, 9)--(1, 2)--(9, 3)--(3, 9);

        \draw[violet] (5, 10)--(5, 5)--(10, 5);
        \draw[violet] (5, 5) node {$\bullet$};
        \draw[violet] (5, 5) node[below] {$\left(\int f, \int g\right)$};

        \draw[red] (6, 6) node {$\bullet$};
        \draw[red] (6, 6) node[above right] {$(u^\star, v^\star)$};

        \draw[orange] (3, 9) node {$\bullet$};
        \draw[orange] (3, 9) node[above] {$(u_1, v_1)$};
        \draw[orange] (9, 3) node {$\bullet$};
        \draw[orange] (9, 3) node[right] {$(u_2, v_2)$};
        \draw[orange] (1, 2) node {$\bullet$};
        \draw[orange] (1, 2) node[below right] {$(u_3, v_3)$};

        \draw[blue] (2, 8) node {$V$};
        \draw[orange] (4, 6) node {$T$};
        \draw[violet] (9, 9) node {$Q$};
	\end{tikzpicture}
	\caption{Applying Carathéodory's theorem} \label{fig:VTQ}
\end{figure}

    Let now $x_1 \in [0, 1]$ be a roll of $D_1$ such that we have $(u_i, v_i) = (f(x_1), g(x_1))$.
    Similarly, let $x_2$ be such that $(u_j, v_j) = (f(x_2), g(x_2))$.
    We define the collective strategy $\bstrat'$ as follows: for $x \leq \lambda$ and for every $y, z$, we define $\strat'_\ada(x, y) = \strat_\ada(x_1, y)$ and $\strat'_\clau(x, z) = \strat_\clau(x_1, z)$. 
    Similarly, for $x > \lambda$, we define $\strat'_\ada(x, y) = \strat_\ada(x_2, y)$ and $\strat'_\clau(x, z) = \strat_\clau(x_2, z)$.
    Bertrand's strategy is left unchanged.

    Intuitively: the reasoning above enables us to choose two \emph{slices} of $\bstrat$, defined by the fixed values $D_1 = x_1$ and $D_1 = x_2$ (see transformation from \cref{fig:sigma0} to \cref{fig:sigma1}).
    The collective strategy $\bstrat'$ is then defined from those two slices by expanding them:the strategy that was played on the slice $D_1 = x_1$ is now the strategy for \emph{all} values of $D_1$ in the interval $[0, \lambda]$, and the strategy that was played on the slice $D_1 = x_2$ is now the strategy for \emph{all} values of $D_1$ in the interval $(\lambda, 1]$.
    As we will show, this guarantees that in this new collective strategy, all team players choose $\H$ with probability $u^\star \geq \int f$, and $\T$ with probability $v^\star \geq \int g$.

    \paragraph*{Some properties of the collective strategy $\bstrat'$}

    The following sublemma is an immediate consequence of the definition of $\bstrat'$.

    \begin{sublemma} \label{prop:2piecewiseconstant}
        The collective strategy $\bstrat'$ is $2$-piecewise constant for the die $D_1$.
    \end{sublemma}

    However, not only does this transformation construct a strategy that is $2$-piecewise constant for the die $D_1$: it also maintains that property for the other dice.
    
    \begin{sublemma}\label{prop:2piecewiseconstantstable}
        If the collective strategy $\bstrat$ is $2$-piecewise constant for some die $D_i$, then so is the collective strategy $\bstrat'$.
    \end{sublemma}

    \begin{proof}
        The case $i = 1$ is immediate by \Cref{prop:2piecewiseconstant}.
        Let us consider $i = 2$; the case $i = 3$ is analogous.
        
        If the collective strategy $\bstrat$ is $2$-piecewise constant for the die $D_2$, it means that the mapping $y \mapsto \bstrat(x, y, z)$ is constant, for every $x$ and $z$.
        Let then $x, z \in [0, 1]$.
        The mapping $y \mapsto \bstrat'(x, y, z)$ is equal to the mapping $y \mapsto \bstrat(x_1, y, z)$ if $x \leq \lambda$, and to the mapping $y \mapsto \bstrat(x_2, y, z)$ otherwise.
        In both case, it is $2$-piecewise constant.
    \end{proof}

    Finally, our transformation preserves the value of our strategy.

\begin{figure}[h]
    \centering
    \begin{subfigure}[t]{0.25\textwidth}
    \centering
	\begin{tikzpicture}[scale=0.18]  
        \draw[->] (0,0) -- (0,10);
        \draw[->] (0,0) -- (10,0);

        \draw (10,0) node[below right] {$x$};
    	\draw (0,10) node[above left] {$y$};

        \fill[red, opacity=0.4] (0,0) -- (4,0) -- (4,3) -- (0,3) -- cycle;
        \fill[red, opacity=0.4] (0,6) -- (4,6) -- (4,9) -- (0,9) -- cycle;
        \fill[blue, opacity=0.4] (4,3) -- (9,3) -- (9,6) -- (4,6) -- cycle;
        \fill[blue, opacity=0.4] (0,3) -- (4,3) -- (4,6) -- (0,6) -- cycle;
        \fill[blue, opacity=0.4] (4,6) -- (9,6) -- (9,9) -- (4,9) -- cycle;
        \fill[blue, opacity=0.4] (4,0) -- (9,0) -- (9,3) -- (4,3) -- cycle;

        \draw[red] (2,1.5) node {$\H$};
        \draw[red] (2,7.5) node {$\H$};
        \draw[blue!70!black] (6.5,4.5) node {$\T$};
    \end{tikzpicture}
    \caption{Ada's strategy}
    \label{fig:sigma1Ada}
    \end{subfigure}
    \begin{subfigure}[t]{0.25\textwidth}
    \centering
	\begin{tikzpicture}[scale=0.18]  
        \draw[->] (0,0) -- (0,10);
        \draw[->] (0,0) -- (10,0);

        \draw (10,0) node[below right] {$z$};
    	\draw (0,10) node[above left] {$y$};

        \fill[red, opacity=0.4] (0,0) -- (3,0) -- (3,3) -- (0,3) -- cycle;
        \fill[red, opacity=0.4] (6,0) -- (9,0) -- (9,3) -- (6,3) -- cycle;
        \fill[red, opacity=0.4] (0,6) -- (3,6) -- (3,9) -- (0,9) -- cycle;
        \fill[red, opacity=0.4] (6,6) -- (9,6) -- (9,9) -- (6,9) -- cycle;
        \fill[blue, opacity=0.4] (3,3) -- (6,3) -- (6,6) -- (3,6) -- cycle;
        \fill[blue, opacity=0.4] (0,3) -- (3,3) -- (3,6) -- (0,6) -- cycle;
        \fill[blue, opacity=0.4] (6,3) -- (9,3) -- (9,6) -- (6,6) -- cycle;
        \fill[blue, opacity=0.4] (3,6) -- (6,6) -- (6,9) -- (3,9) -- cycle;
        \fill[blue, opacity=0.4] (3,0) -- (6,0) -- (6,3) -- (3,3) -- cycle;

        \draw[red] (1.5,1.5) node {$\H$};
        \draw[red] (7.5,7.5) node {$\H$};
        \draw[red] (7.5,1.5) node {$\H$};
        \draw[red] (1.5,7.5) node {$\H$};
        \draw[blue!70!black] (4.5,4.5) node {$\T$};
    \end{tikzpicture}
    \caption{Bertrand's strategy}
    \label{fig:sigma1Bertrand}
    \end{subfigure}
    \begin{subfigure}[t]{0.25\textwidth}
    \centering
	\begin{tikzpicture}[scale=0.18]  
        \draw[->] (0,0) -- (0,10);
        \draw[->] (0,0) -- (10,0);

        \draw (10,0) node[below right] {$x$};
    	\draw (0,10) node[above left] {$z$};

        \fill[red, opacity=0.4] (0,0) -- (4,0) -- (4,3) -- (0,3) -- cycle;
        \fill[red, opacity=0.4] (0,6) -- (4,6) -- (4,9) -- (0,9) -- cycle;
        \fill[blue, opacity=0.4] (4,3) -- (9,3) -- (9,6) -- (4,6) -- cycle;
        \fill[blue, opacity=0.4] (0,3) -- (4,3) -- (4,6) -- (0,6) -- cycle;
        \fill[blue, opacity=0.4] (4,6) -- (9,6) -- (9,9) -- (4,9) -- cycle;
        \fill[blue, opacity=0.4] (4,0) -- (9,0) -- (9,3) -- (4,3) -- cycle;

        \draw[red] (2,1.5) node {$\H$};
        \draw[red] (2,7.5) node {$\H$};
        \draw[blue!70!black] (6.5,4.5) node {$\T$};
    \end{tikzpicture}
    \caption{Claude's strategy}
    \label{fig:sigma1Claude}
    \end{subfigure}
    \begin{subfigure}[t]{0.35\textwidth}
    \centering
	\begin{tikzpicture}[scale=0.25, line join=round]  
                
        \newcommand{\cuboidr}[6]{%
          \pgfmathsetmacro{\x}{#1}
          \pgfmathsetmacro{\y}{#2}
          \pgfmathsetmacro{\z}{#3}
          \pgfmathsetmacro{\xa}{\x+#4}
          \pgfmathsetmacro{\ya}{\y+#5}
          \pgfmathsetmacro{\za}{\z+#6}
        
          \fill[red, opacity=0.4] (\x,\y,\z) -- (\xa,\y,\z) -- (\xa,\ya,\z) -- (\x,\ya,\z) -- cycle; 
          \fill[red, opacity=0.4] (\x,\y,\z) -- (\xa,\y,\z) -- (\xa,\y,\za) -- (\x,\y,\za) -- cycle; 
          \fill[red, opacity=0.4] (\x,\y,\z) -- (\x,\ya,\z) -- (\x,\ya,\za) -- (\x,\y,\za) -- cycle; 
          \fill[red, opacity=0.4] (\xa,\ya,\za) -- (\x,\ya,\za) -- (\x,\y,\za) -- (\xa,\y,\za) -- cycle; 
          \fill[red, opacity=0.4] (\xa,\ya,\za) -- (\xa,\y,\za) -- (\xa,\y,\z) -- (\xa,\ya,\z) -- cycle; 
          \fill[red, opacity=0.4] (\xa,\ya,\za) -- (\xa,\ya,\z) -- (\x,\ya,\z) -- (\x,\ya,\za) -- cycle; 
        
          \draw[thick] (\x,\y,\z) -- (\xa,\y,\z) -- (\xa,\ya,\z) -- (\x,\ya,\z) -- cycle; 
          \draw[thick] (\x,\y,\za) -- (\xa,\y,\za) -- (\xa,\ya,\za) -- (\x,\ya,\za) -- cycle; 
          \draw[thick] (\x,\y,\z) -- (\x,\y,\za);
          \draw[thick] (\xa,\y,\z) -- (\xa,\y,\za);
          \draw[thick] (\xa,\ya,\z) -- (\xa,\ya,\za);
          \draw[thick] (\x,\ya,\z) -- (\x,\ya,\za);
        }

        \newcommand{\cuboidg}[6]{%
          \pgfmathsetmacro{\x}{#1}
          \pgfmathsetmacro{\y}{#2}
          \pgfmathsetmacro{\z}{#3}
          \pgfmathsetmacro{\xa}{\x+#4}
          \pgfmathsetmacro{\ya}{\y+#5}
          \pgfmathsetmacro{\za}{\z+#6}
        
          \fill[blue, opacity=0.4] (\x,\y,\z) -- (\xa,\y,\z) -- (\xa,\ya,\z) -- (\x,\ya,\z) -- cycle; 
          \fill[blue, opacity=0.4] (\x,\y,\z) -- (\xa,\y,\z) -- (\xa,\y,\za) -- (\x,\y,\za) -- cycle; 
          \fill[blue, opacity=0.4] (\x,\y,\z) -- (\x,\ya,\z) -- (\x,\ya,\za) -- (\x,\y,\za) -- cycle; 
          \fill[blue, opacity=0.4] (\xa,\ya,\za) -- (\x,\ya,\za) -- (\x,\y,\za) -- (\xa,\y,\za) -- cycle; 
          \fill[blue, opacity=0.4] (\xa,\ya,\za) -- (\xa,\y,\za) -- (\xa,\y,\z) -- (\xa,\ya,\z) -- cycle; 
          \fill[blue, opacity=0.4] (\xa,\ya,\za) -- (\xa,\ya,\z) -- (\x,\ya,\z) -- (\x,\ya,\za) -- cycle; 
        
          \draw[thick] (\x,\y,\z) -- (\xa,\y,\z) -- (\xa,\ya,\z) -- (\x,\ya,\z) -- cycle; 
          \draw[thick] (\x,\y,\za) -- (\xa,\y,\za) -- (\xa,\ya,\za) -- (\x,\ya,\za) -- cycle; 
          \draw[thick] (\x,\y,\z) -- (\x,\y,\za);
          \draw[thick] (\xa,\y,\z) -- (\xa,\y,\za);
          \draw[thick] (\xa,\ya,\z) -- (\xa,\ya,\za);
          \draw[thick] (\x,\ya,\z) -- (\x,\ya,\za);
        }

        \draw[->] (0,0,0) -- (10,0,0);
    	\draw[->] (0,0,0) -- (0,10,0);
    	\draw[->] (0,0,0) -- (0,0,10);
    	
    	\draw (10,0,0) node[below right] {$x$};
    	\draw (0,10,0) node[above left] {$y$};
    	\draw (0,0,10) node[left] {$z$};

        \fill[black, opacity=0.3] (0,1,0)--(0,1,9)--(9,1,9)--(9,1,0);

        \cuboidr{0}{0}{0}{4}{3}{3}
        \cuboidr{0}{0}{6}{4}{3}{3}

        \draw[very thick, black, dashed] (0,1,0)--(0,1,3)--(4,1,3)--(4,1,0)--cycle;
        \draw[very thick, black, dashed] (0,1,6)--(0,1,9)--(4,1,9)--(4,1,6)--cycle;

        \fill[black, opacity=0.3] (0,4,0)--(0,4,6)--(9,4,6)--(9,4,0);

        \cuboidg{4}{3}{3}{5}{3}{3}

        \fill[black, opacity=0.3] (0,4,6)--(0,4,9)--(9,4,9)--(9,4,6);
        
        \cuboidr{0}{6}{0}{4}{3}{3}
        \cuboidr{0}{6}{6}{4}{3}{3}

        \draw[very thick, black, dashed] (4,4,3)--(4,4,6)--(9,4,6)--(9,4,3)--cycle;
    \end{tikzpicture}
    \caption{The sets $\H\H\H$ and $\T\T\T$}
    \label{fig:sigma1HHH}
    \end{subfigure}
    \begin{subfigure}[t]{0.2\linewidth}
    \centering
	\begin{tikzpicture}[scale=0.1]  
        \draw[->] (0,0) -- (0,10);
        \draw[->] (0,0) -- (10,0);

        \draw (10,0) node[below right] {$x$};
    	\draw (0,10) node[above left] {$z$};

        \draw[very thick, black] (0,0)--(0,9)--(9,9)--(9,0)--cycle;

        \fill[red, opacity=0.4] (0,0) -- (4,0) -- (4,3) -- (0,3) -- cycle;
        \fill[red, opacity=0.4] (0,6) -- (4,6) -- (4,9) -- (0,9) -- cycle;
    \end{tikzpicture}
    \caption{The slice $y=0.1$}
    \label{fig:sigma1y=1}
    \end{subfigure}
    \begin{subfigure}[t]{0.2\linewidth}
    \centering
	\begin{tikzpicture}[scale=0.1]  
        \draw[->] (0,0) -- (0,10);
        \draw[->] (0,0) -- (10,0);

        \draw (10,0) node[below right] {$x$};
    	\draw (0,10) node[above left] {$z$};

        \draw[very thick, black] (0,0)--(0,9)--(9,9)--(9,0)--cycle;
        
        \fill[blue, opacity=0.4] (4,3) -- (9,3) -- (9,6) -- (4,6) -- cycle;
    \end{tikzpicture}
    \caption{The slice $y=0.4$}
    \label{fig:sigma1y=4}
    \end{subfigure}
	\caption{The collective strategy $\phi_1(\bstrat^0)$} \label{fig:sigma1}
\end{figure}

\begin{figure}[h]
    \centering
    \begin{subfigure}[t]{0.25\textwidth}
    \centering
	\begin{tikzpicture}[scale=0.18]  
        \draw[->] (0,0) -- (0,10);
        \draw[->] (0,0) -- (10,0);

        \draw (10,0) node[below right] {$x$};
    	\draw (0,10) node[above left] {$y$};

        \fill[red, opacity=0.4] (0,0) -- (4,0) -- (4,7) -- (0,7) -- cycle;
        \fill[blue, opacity=0.4] (4,7) -- (9,7) -- (9,9) -- (4,9) -- cycle;
        \fill[blue, opacity=0.4] (0,7) -- (4,7) -- (4,9) -- (0,9) -- cycle;
        \fill[blue, opacity=0.4] (4,0) -- (9,0) -- (9,7) -- (4,7) -- cycle;

        \draw[red] (2,3.5) node {$\H$};
        \draw[blue!70!black] (6.5,8) node {$\T$};
    \end{tikzpicture}
    \caption{Ada's strategy}
    \label{fig:sigma2Ada}
    \end{subfigure}
    \begin{subfigure}[t]{0.25\textwidth}
    \centering
	\begin{tikzpicture}[scale=0.18]  
        \draw[->] (0,0) -- (0,10);
        \draw[->] (0,0) -- (10,0);

        \draw (10,0) node[below right] {$z$};
    	\draw (0,10) node[above left] {$y$};

        \fill[red, opacity=0.4] (0,0) -- (3,0) -- (3,7) -- (0,7) -- cycle;
        \fill[red, opacity=0.4] (6,0) -- (9,0) -- (9,7) -- (6,7) -- cycle;
        \fill[blue, opacity=0.4] (3,7) -- (6,7) -- (6,9) -- (3,9) -- cycle;
        \fill[blue, opacity=0.4] (0,7) -- (3,7) -- (3,9) -- (0,9) -- cycle;
        \fill[blue, opacity=0.4] (6,7) -- (9,7) -- (9,9) -- (6,9) -- cycle;
        \fill[blue, opacity=0.4] (3,0) -- (6,0) -- (6,7) -- (3,7) -- cycle;

        \draw[red] (1.5,3.5) node {$\H$};
        \draw[blue!70!black] (4.5,8) node {$\T$};
        \draw[red] (7.5,3.5) node {$\H$};
    \end{tikzpicture}
    \caption{Bertrand's strategy}
    \label{fig:sigma2Bertrand}
    \end{subfigure}
    \begin{subfigure}[t]{0.25\textwidth}
    \centering
	\begin{tikzpicture}[scale=0.18]  
        \draw[->] (0,0) -- (0,10);
        \draw[->] (0,0) -- (10,0);

        \draw (10,0) node[below right] {$x$};
    	\draw (0,10) node[above left] {$z$};

        \fill[red, opacity=0.4] (0,0) -- (4,0) -- (4,3) -- (0,3) -- cycle;
        \fill[red, opacity=0.4] (0,6) -- (4,6) -- (4,9) -- (0,9) -- cycle;
        \fill[blue, opacity=0.4] (4,3) -- (9,3) -- (9,6) -- (4,6) -- cycle;
        \fill[blue, opacity=0.4] (0,3) -- (4,3) -- (4,6) -- (0,6) -- cycle;
        \fill[blue, opacity=0.4] (4,6) -- (9,6) -- (9,9) -- (4,9) -- cycle;
        \fill[blue, opacity=0.4] (4,0) -- (9,0) -- (9,3) -- (4,3) -- cycle;

        \draw[red] (2,1.5) node {$\H$};
        \draw[red] (2,7.5) node {$\H$};
        \draw[blue!70!black] (6.5,4.5) node {$\T$};
    \end{tikzpicture}
    \caption{Claude's strategy}
    \label{fig:sigma2Claude}
    \end{subfigure}
    \begin{subfigure}[t]{0.35\textwidth}
    \centering
	\begin{tikzpicture}[scale=0.25, line join=round]  
                
        \newcommand{\rcuboid}[6]{%
          \pgfmathsetmacro{\x}{#1}
          \pgfmathsetmacro{\y}{#2}
          \pgfmathsetmacro{\z}{#3}
          \pgfmathsetmacro{\xa}{\x+#4}
          \pgfmathsetmacro{\ya}{\y+#5}
          \pgfmathsetmacro{\za}{\z+#6}
        
          \fill[red, opacity=0.4] (\x,\y,\z) -- (\xa,\y,\z) -- (\xa,\ya,\z) -- (\x,\ya,\z) -- cycle; 
          \fill[red, opacity=0.4] (\x,\y,\z) -- (\xa,\y,\z) -- (\xa,\y,\za) -- (\x,\y,\za) -- cycle; 
          \fill[red, opacity=0.4] (\x,\y,\z) -- (\x,\ya,\z) -- (\x,\ya,\za) -- (\x,\y,\za) -- cycle; 
          \fill[red, opacity=0.4] (\xa,\ya,\za) -- (\x,\ya,\za) -- (\x,\y,\za) -- (\xa,\y,\za) -- cycle; 
          \fill[red, opacity=0.4] (\xa,\ya,\za) -- (\xa,\y,\za) -- (\xa,\y,\z) -- (\xa,\ya,\z) -- cycle; 
          \fill[red, opacity=0.4] (\xa,\ya,\za) -- (\xa,\ya,\z) -- (\x,\ya,\z) -- (\x,\ya,\za) -- cycle; 
        
          \draw[thick] (\x,\y,\z) -- (\xa,\y,\z) -- (\xa,\ya,\z) -- (\x,\ya,\z) -- cycle; 
          \draw[thick] (\x,\y,\za) -- (\xa,\y,\za) -- (\xa,\ya,\za) -- (\x,\ya,\za) -- cycle; 
          \draw[thick] (\x,\y,\z) -- (\x,\y,\za);
          \draw[thick] (\xa,\y,\z) -- (\xa,\y,\za);
          \draw[thick] (\xa,\ya,\z) -- (\xa,\ya,\za);
          \draw[thick] (\x,\ya,\z) -- (\x,\ya,\za);
        }

        \newcommand{\gcuboid}[6]{%
          \pgfmathsetmacro{\x}{#1}
          \pgfmathsetmacro{\y}{#2}
          \pgfmathsetmacro{\z}{#3}
          \pgfmathsetmacro{\xa}{\x+#4}
          \pgfmathsetmacro{\ya}{\y+#5}
          \pgfmathsetmacro{\za}{\z+#6}
        
          \fill[blue, opacity=0.4] (\x,\y,\z) -- (\xa,\y,\z) -- (\xa,\ya,\z) -- (\x,\ya,\z) -- cycle; 
          \fill[blue, opacity=0.4] (\x,\y,\z) -- (\xa,\y,\z) -- (\xa,\y,\za) -- (\x,\y,\za) -- cycle; 
          \fill[blue, opacity=0.4] (\x,\y,\z) -- (\x,\ya,\z) -- (\x,\ya,\za) -- (\x,\y,\za) -- cycle; 
          \fill[blue, opacity=0.4] (\xa,\ya,\za) -- (\x,\ya,\za) -- (\x,\y,\za) -- (\xa,\y,\za) -- cycle; 
          \fill[blue, opacity=0.4] (\xa,\ya,\za) -- (\xa,\y,\za) -- (\xa,\y,\z) -- (\xa,\ya,\z) -- cycle; 
          \fill[blue, opacity=0.4] (\xa,\ya,\za) -- (\xa,\ya,\z) -- (\x,\ya,\z) -- (\x,\ya,\za) -- cycle; 
        
          \draw[thick] (\x,\y,\z) -- (\xa,\y,\z) -- (\xa,\ya,\z) -- (\x,\ya,\z) -- cycle; 
          \draw[thick] (\x,\y,\za) -- (\xa,\y,\za) -- (\xa,\ya,\za) -- (\x,\ya,\za) -- cycle; 
          \draw[thick] (\x,\y,\z) -- (\x,\y,\za);
          \draw[thick] (\xa,\y,\z) -- (\xa,\y,\za);
          \draw[thick] (\xa,\ya,\z) -- (\xa,\ya,\za);
          \draw[thick] (\x,\ya,\z) -- (\x,\ya,\za);
        }

        \draw[->] (0,0,0) -- (10,0,0);
    	\draw[->] (0,0,0) -- (0,10,0);
    	\draw[->] (0,0,0) -- (0,0,10);
    	
    	\draw (10,0,0) node[below right] {$x$};
    	\draw (0,10,0) node[above left] {$y$};
    	\draw (0,0,10) node[left] {$z$};

        \rcuboid{0}{0}{0}{4}{7}{3}

        \fill[black, opacity=0.3] (0,0,4)--(0,9,4)--(9,9,4)--(9,0,4);

        \gcuboid{4}{7}{3}{5}{2}{3}

        \draw[very thick, black, dashed] (4,7,4)--(4,9,4)--(9,9,4)--(9,7,4)--cycle;

        \fill[black, opacity=0.3] (0,0,8)--(0,9,8)--(9,9,8)--(9,0,8);
        
        \rcuboid{0}{0}{6}{4}{7}{3}
        
        \draw[very thick, black, dashed] (0,0,8)--(0,7,8)--(4,7,8)--(4,0,8)--cycle;
    \end{tikzpicture}
    \caption{The sets $\H\H\H$ and $\T\T\T$}
    \label{fig:sigma2HHH}
    \end{subfigure}
    \begin{subfigure}[t]{0.2\linewidth}
    \centering
	\begin{tikzpicture}[scale=0.1]  
        \draw[->] (0,0) -- (0,10);
        \draw[->] (0,0) -- (10,0);

        \draw (10,0) node[below right] {$x$};
    	\draw (0,10) node[above left] {$z$};

        \draw[very thick, black] (0,0)--(0,9)--(9,9)--(9,0)--cycle;

        \fill[blue, opacity=0.4] (4,7)--(4,9)--(9,9)--(9,7)--cycle;
    \end{tikzpicture}
    \caption{The slice $z=0.4$}
    \label{fig:sigma2z=4}
    \end{subfigure}
    \begin{subfigure}[t]{0.2\linewidth}
    \centering
	\begin{tikzpicture}[scale=0.1]  
        \draw[->] (0,0) -- (0,10);
        \draw[->] (0,0) -- (10,0);

        \draw (10,0) node[below right] {$x$};
    	\draw (0,10) node[above left] {$z$};

        \draw[very thick, black] (0,0)--(0,9)--(9,9)--(9,0)--cycle;
        
        \fill[red, opacity=0.4] (0,0)--(0,7)--(4,7)--(4,0)--cycle;
    \end{tikzpicture}
    \caption{The slice $z=0.8$}
    \label{fig:sigma2z=8}
    \end{subfigure}
	\caption{The collective strategy $\phi_2 \circ \phi_1(\bstrat^0)$} \label{fig:sigma2}
\end{figure}

\begin{figure}[h]
    \centering
    \begin{subfigure}[t]{0.22\textwidth}
	\begin{tikzpicture}[scale=0.18]  
        \draw[->] (0,0) -- (0,10);
        \draw[->] (0,0) -- (10,0);

        \draw (10,0) node[below right] {$x$};
    	\draw (0,10) node[above left] {$y$};
˜
        \fill[red, opacity=0.4] (0,0) -- (4,0) -- (4,7) -- (0,7) -- cycle;
        \fill[blue, opacity=0.4] (4,7) -- (9,7) -- (9,9) -- (4,9) -- cycle;
        \fill[blue, opacity=0.4] (0,7) -- (4,7) -- (4,9) -- (0,9) -- cycle;
        \fill[blue, opacity=0.4] (4,0) -- (9,0) -- (9,7) -- (4,7) -- cycle;

        \draw[red] (2,3.5) node {$\H$};
        \draw[blue!70!black] (6.5,8) node {$\T$};
    \end{tikzpicture}
    \caption{Ada's strategy}
    \label{fig:sigmastarAda}
    \end{subfigure}
    \begin{subfigure}[t]{0.22\textwidth}
    \centering
	\begin{tikzpicture}[scale=0.18]  
        \draw[->] (0,0) -- (0,10);
        \draw[->] (0,0) -- (10,0);

        \draw (10,0) node[below right] {$z$};
    	\draw (0,10) node[above left] {$y$};

        \fill[red, opacity=0.4] (5,0) -- (9,0) -- (9,7) -- (5,7) -- cycle;
        \fill[blue, opacity=0.4] (0,7) -- (5,7) -- (5,9) -- (0,9) -- cycle;
        \fill[blue, opacity=0.4] (5,7) -- (9,7) -- (9,9) -- (5,9) -- cycle;
        \fill[blue, opacity=0.4] (0,0) -- (5,0) -- (5,7) -- (0,7) -- cycle;

        \draw[blue!70!black] (2.5,8) node {$\T$};
        \draw[red] (7,3.5) node {$\H$};
    \end{tikzpicture}
    \caption{Bertrand's strategy}
    \label{fig:sigmastarBertrand}
    \end{subfigure}
    \begin{subfigure}[t]{0.22\textwidth}
    \centering
	\begin{tikzpicture}[scale=0.18]  
        \draw[->] (0,0) -- (0,10);
        \draw[->] (0,0) -- (10,0);

        \draw (10,0) node[below right] {$x$};
    	\draw (0,10) node[above left] {$z$};

        \fill[red, opacity=0.4] (0,5) -- (4,5) -- (4,9) -- (0,9) -- cycle;
        \fill[blue, opacity=0.4] (4,0) -- (9,0) -- (9,5) -- (4,5) -- cycle;
        \fill[blue, opacity=0.4] (0,0) -- (4,0) -- (4,5) -- (0,5) -- cycle;
        \fill[blue, opacity=0.4] (4,5) -- (9,5) -- (9,9) -- (4,9) -- cycle;

        \draw[red] (2,7) node {$\H$};
        \draw[blue!70!black] (6.5,2.5) node {$\T$};
    \end{tikzpicture}
    \caption{Claude's strategy}
    \label{fig:sigmastarClaude}
    \end{subfigure}
    \begin{subfigure}[t]{0.25\textwidth}
    \centering
	\begin{tikzpicture}[scale=0.2, line join=round]  
                
        \newcommand{\rcuboid}[6]{
          \pgfmathsetmacro{\x}{#1}
          \pgfmathsetmacro{\y}{#2}
          \pgfmathsetmacro{\z}{#3}
          \pgfmathsetmacro{\xa}{\x+#4}
          \pgfmathsetmacro{\ya}{\y+#5}
          \pgfmathsetmacro{\za}{\z+#6}
        
          \fill[red, opacity=0.4] (\x,\y,\z) -- (\xa,\y,\z) -- (\xa,\ya,\z) -- (\x,\ya,\z) -- cycle; 
          \fill[red, opacity=0.4] (\x,\y,\z) -- (\xa,\y,\z) -- (\xa,\y,\za) -- (\x,\y,\za) -- cycle; 
          \fill[red, opacity=0.4] (\x,\y,\z) -- (\x,\ya,\z) -- (\x,\ya,\za) -- (\x,\y,\za) -- cycle; 
          \fill[red, opacity=0.4] (\xa,\ya,\za) -- (\x,\ya,\za) -- (\x,\y,\za) -- (\xa,\y,\za) -- cycle; 
          \fill[red, opacity=0.4] (\xa,\ya,\za) -- (\xa,\y,\za) -- (\xa,\y,\z) -- (\xa,\ya,\z) -- cycle; 
          \fill[red, opacity=0.4] (\xa,\ya,\za) -- (\xa,\ya,\z) -- (\x,\ya,\z) -- (\x,\ya,\za) -- cycle; 
        
          \draw[thick] (\x,\y,\z) -- (\xa,\y,\z) -- (\xa,\ya,\z) -- (\x,\ya,\z) -- cycle; 
          \draw[thick] (\x,\y,\za) -- (\xa,\y,\za) -- (\xa,\ya,\za) -- (\x,\ya,\za) -- cycle; 
          \draw[thick] (\x,\y,\z) -- (\x,\y,\za);
          \draw[thick] (\xa,\y,\z) -- (\xa,\y,\za);
          \draw[thick] (\xa,\ya,\z) -- (\xa,\ya,\za);
          \draw[thick] (\x,\ya,\z) -- (\x,\ya,\za);
        }

        \newcommand{\gcuboid}[6]{%
          \pgfmathsetmacro{\x}{#1}
          \pgfmathsetmacro{\y}{#2}
          \pgfmathsetmacro{\z}{#3}
          \pgfmathsetmacro{\xa}{\x+#4}
          \pgfmathsetmacro{\ya}{\y+#5}
          \pgfmathsetmacro{\za}{\z+#6}
        
          \fill[blue, opacity=0.4] (\x,\y,\z) -- (\xa,\y,\z) -- (\xa,\ya,\z) -- (\x,\ya,\z) -- cycle; 
          \fill[blue, opacity=0.4] (\x,\y,\z) -- (\xa,\y,\z) -- (\xa,\y,\za) -- (\x,\y,\za) -- cycle; 
          \fill[blue, opacity=0.4] (\x,\y,\z) -- (\x,\ya,\z) -- (\x,\ya,\za) -- (\x,\y,\za) -- cycle; 
          \fill[blue, opacity=0.4] (\xa,\ya,\za) -- (\x,\ya,\za) -- (\x,\y,\za) -- (\xa,\y,\za) -- cycle; 
          \fill[blue, opacity=0.4] (\xa,\ya,\za) -- (\xa,\y,\za) -- (\xa,\y,\z) -- (\xa,\ya,\z) -- cycle; 
          \fill[blue, opacity=0.4] (\xa,\ya,\za) -- (\xa,\ya,\z) -- (\x,\ya,\z) -- (\x,\ya,\za) -- cycle; 
        
          \draw[thick] (\x,\y,\z) -- (\xa,\y,\z) -- (\xa,\ya,\z) -- (\x,\ya,\z) -- cycle; 
          \draw[thick] (\x,\y,\za) -- (\xa,\y,\za) -- (\xa,\ya,\za) -- (\x,\ya,\za) -- cycle; 
          \draw[thick] (\x,\y,\z) -- (\x,\y,\za);
          \draw[thick] (\xa,\y,\z) -- (\xa,\y,\za);
          \draw[thick] (\xa,\ya,\z) -- (\xa,\ya,\za);
          \draw[thick] (\x,\ya,\z) -- (\x,\ya,\za);
        }

        \draw[->] (0,0,0) -- (10,0,0);
    	\draw[->] (0,0,0) -- (0,10,0);
    	\draw[->] (0,0,0) -- (0,0,10);
    	
    	\draw (10,0,0) node[below right] {$x$};
    	\draw (0,10,0) node[above left] {$y$};
    	\draw (0,0,10) node[left] {$z$};

        \gcuboid{4}{7}{0}{5}{2}{5} 
        \rcuboid{0}{0}{5}{4}{7}{4}
    \end{tikzpicture}
    \caption{The sets $\H\H\H$ and $\T\T\T$}
    \label{fig:sigmastarHHH}
    \end{subfigure}
	\caption{The collective strategy $\bstrat^\star$} \label{fig:sigmastar}
\end{figure}
    \begin{sublemma} \label{prop:maintainsvalue}
        The collective strategy $\bstrat'$ has a value at least as large as the collective strategy $\bstrat$.
    \end{sublemma}

    \begin{proof}
        We prove this sublemma by showing the inequalities:
        $$\prob\left( \strat'_\ada = \strat'_\bert = \strat'_\clau = \H \right) \geq \int f\quad\text{and}\quad\prob\left( \strat'_\ada = \strat'_\bert = \strat'_\clau = \T \right) \geq \int g.$$
        Let us prove the first one:
        \begin{align*}
            &\prob\left( \strat'_\ada = \strat'_\bert = \strat'_\clau = \H \right) \\
            =~ & \prob(D_1 \leq \lambda) \prob\left( \strat'_\ada = \strat'_\bert = \strat'_\clau = \H \mid D_1 \leq \lambda \right)\\
            &+ \prob(D_1 > \lambda) \prob\left( \strat'_\ada = \strat'_\bert = \strat'_\clau = \H \mid D_1 > \lambda \right) \\
            =~ & \lambda \prob\left( \strat_\ada = \strat_\bert = \strat_\clau = \H \mid D_1 = x_1 \right) \\
            & + (1-\lambda) \prob\left( \strat_\ada = \strat_\bert = \strat_\clau = \H \mid D_1 = x_2 \right) \\
        =~ & \lambda f(x_1) + (1-\lambda) f(x_2) \\
        =~ & u^\star \geq \int f.
        \end{align*}
        The second one is analogous.

        Then, the value of the collective strategy $\bstrat'$ is:
        $$\min\left\{ \prob\left( \strat'_\ada = \strat'_\bert = \strat'_\clau = \H \right), \prob\left( \strat'_\ada = \strat'_\bert = \strat'_\clau = \T\right) \right\}$$
        $$\geq \min\left\{\int f, \int g\right\},$$
        which was the value of the collective strategy $\bstrat$.
    \end{proof}

    \paragraph*{Conclusion}

    Let us write $\phi_1: \bstrat \mapsto \bstrat'$ for the transformation described above.
    We define analogously the transformations $\phi_2$ and $\phi_3$ that construct new collective strategies from two slices taken by fixing the second and the third roll, respectively, instead of the first one.
    The same proofs show that these transformations also satisfy \Cref{prop:2piecewiseconstant} (with $D_2$ and $D_3$, respectively, instead of $D_1$), \Cref{prop:2piecewiseconstantstable} and \Cref{prop:maintainsvalue}.

    From the collective strategy $\bstrat^0$, we then define the collective strategy $\bstrat^\star = \phi_3 \circ \phi_2 \circ \phi_1(\bstrat^0)$.
    By \Cref{prop:maintainsvalue}, the value of this new collective strategy is at least as high as that of $\bstrat^0$, and by \Cref{prop:2piecewiseconstant} and \Cref{prop:2piecewiseconstantstable}, it is $2$-piecewise constant for dice $D_1, D_2$, and $D_3$.
    In other words, it is $2$-grid.
\end{proof}

A consequence of \Cref{lm:itsOKtobestraightexample} is that the supremum of values of collective strategies in this game is the supremum of quantities $t$ satisfying the following inequations:
\begin{equation}
\begin{split}
    t \leq &~\lambda_1 \lambda_2 \lambda_3 a_{11} b_{ 11} c_{ 11}\\
    &+ \lambda_1 \lambda_2 (1-\lambda_3) a_{11} b_{ 12} c_{ 12} + \dots \\
    &+ (1-\lambda_1) (1-\lambda_2) (1-\lambda_3) a_{22} b_{ 22} c_{ 22}
\end{split}
\end{equation}

\begin{equation}
\begin{split}
    t \leq &~\lambda_1 \lambda_2 \lambda_3 (1-a_{11}) (1-b_{ 11}) (1-c_{ 11}) + \dots \\
    &+ (1-\lambda_1) (1-\lambda_2) (1-\lambda_3) (1-a_{22}) (1-b_{ 22}) (1-c_{ 22})
\end{split}
\end{equation}
\begin{equation}
    \forall i \in \{1, 2, 3\}, 0 \leq \lambda_i \leq 1
\end{equation}
(the full version of this system is available in App.~\ref{app:inequations}), where $t, \lambda_1, \lambda_2$, and $\lambda_3$ are real variables, representing the thresholds chosen for dice $D_1, D_2$, and $D_3$, respectively, and $a_{ij}, b_{ jk},$ and $c_{ ik}$ are Boolean variables, such that $a_{ij} = 0$ or $a_{ij} = 1$ means that Ada plays Heads or Tails, respectively, when $D_1$ takes a value that is smaller than or equal to (if $i = 1$) the threshold $\lambda_1$, or greater than (if $i = 2)$ that threshold, and $D_2$ is smaller than or equal to (if $j = 1$) the threshold $\lambda_2$, or greater than (if if $j = 2)$ that threshold.

Since all inequalities are non-strict, this system has a solution that maximises the variable $t$; using the SMT solver Z3~\cite{Zthree}, we find that that maximal value for $t$ is $\beta$, which proves \Cref{thm:3v1example}.

\subsection{General result: the optimality of grid strategies}

We now generalise \Cref{lm:itsOKtobestraightexample}   to all dicey games.
The proof follows the same structure.

\begin{restatable}[App.~\ref{app:itsOKtobestraight}]{theorem}{straightLineTheorem}\label{thm:itsOKtobestraight}
    Let $\DGame$ be a dicey game.
    Let $k$ be the number of actions available for the Devil.
    Then, for every collective strategy, there exists a $k$-grid strategy whose value is at least as large.
\end{restatable}

\section{Optimal strategies and their size}\label{sec:optimal}
In this section, we build on \Cref{thm:itsOKtobestraight} to give concrete tools for computing collective strategies in dicey games.

Thanks to \Cref{thm:itsOKtobestraight}, the existence of a collective strategy achieving a given threshold payoff $t$ can be expressed as a sentence in the existential theory of the reals (ETR). Concretely, this provides a system of exponentially many polynomial inequalities over polynomially many existentially quantified real variables. This formulation is already useful for computational purposes, and we spell it out explicitly in \Cref{lm:inequations}.

The main result of the section 
goes a step further: we show that optimal collective strategies exist, and admit finite descriptions whose size can be bounded. The proof proceeds in several stages.

By \Cref{lm:inequations}, achievable payoffs are characterised by a system of inequations, and the $\max\min$ value is the largest number $t$ for which the solution set of the system is nonempty.
To capture optimality, we introduce the Fritz John optimality conditions~\cite{John48,mangasarian67FritzJohn}, which provide necessary conditions for $t$ to attain its maximum.
These conditions define the semi-algebraic set of \emph{Fritz John's points}.



We then appeal to a result from real algebraic geometry due to Collins, from his work on cylindrical algebraic decomposition~\cite{Col74}. This result implies that the set of Fritz John's points, like any semi-algebraic set, can be partitioned into finitely many connected components, called \emph{strata}.
A key step is then to show that within each stratum, the value of $t$ is constant (\cref{prop:FritzJohnfinite}). The set of Fritz John's points has therefore the shape of a French \emph{millefeuille}, made of finitely many flat strata: see \Cref{fig:millefeuille} for an illustration.
Finally, we invoke results from computational real algebraic geometry due to Basu, Pollack, and Roy~\cite{BPR96,BPR97}, which guarantee that each connected component (and thus each stratum) contains a point with at most exponential bit complexity.
This is in particular the case of the top stratum, the one meant to hold the frosting: this implies that both the optimal value and the probabilities defining an optimal collective strategy admit exponential size representations.

\begin{figure}[h]
    \centering
\begin{tikzpicture}[
    x={(1cm,0cm)}, 
    y={(0cm,1cm)}, 
    z={(0.5cm,0.3cm)}, 
    scale=0.8,
    font=\sffamily\small
]

\def\width{3.5}
\def\depth{3.0}
\def\hPastry{0.15} 
\def\hCream{0.5}  
\def\hFrost{0.3}  

\definecolor{pastryColor}{HTML}{D2691E} 
\definecolor{pastryFace}{HTML}{DEB887}  
\definecolor{creamColor}{HTML}{FFFACD}   
\definecolor{creamFace}{HTML}{F0E68C}    
\definecolor{frostingColor}{HTML}{FFFAF0} 
\definecolor{chocolateLine}{HTML}{4B3621} 

\pgfmathsetmacro{\yPos}{0}

\newcommand{\drawLayer}[5]{
    \pgfmathsetmacro{\tH}{#1 + #2} 
    \pgfmathsetmacro{\mH}{#1 + #2/2} 

    \fill[#3] (0, \tH, 0) -- (\width, \tH, 0) -- (\width, \tH, \depth) -- (0, \tH, \depth) -- cycle;
    
    \draw[thin, black!20] (0, \tH, \depth) -- (\width, \tH, \depth) -- (\width, \tH, 0); 
    \draw[thin, black!20] (0, #1, \depth) -- (0, \tH, \depth); 
    
    \node[anchor=east] at (0, \mH, \depth) {#5};
}


\pgfmathsetmacro{\totalHeight}{2*\hPastry + 2*\hCream + \hFrost}
\draw[->, thick, >=stealth] (-1, -0.5, \depth) -- (-1, \totalHeight, \depth) node[above] {$t$};
\draw (-1.1, 0, \depth) -- (-0.9, 0, \depth); 

\drawLayer{\yPos}{\hPastry}{pastryFace}{pastryColor}{$S_1$}
\pgfmathsetmacro{\yPos}{\yPos + \hPastry}

\drawLayer{\yPos}{\hCream}{creamColor}{creamFace}{$S_2$}
\pgfmathsetmacro{\yPos}{\yPos + \hCream}

\drawLayer{\yPos}{\hPastry}{pastryFace}{pastryColor}{$S_3$}
\pgfmathsetmacro{\yPos}{\yPos + \hPastry}

\drawLayer{\yPos}{\hCream}{creamColor}{creamFace}{$S_4$}
\pgfmathsetmacro{\yPos}{\yPos + \hCream}

\pgfmathsetmacro{\fStart}{\yPos}
\pgfmathsetmacro{\fTop}{\yPos + \hFrost}

\drawLayer{\fStart}{\hFrost}{frostingColor}{frostingColor!90!gray}{}

\begin{scope}[canvas is xz plane at y=\fTop]
    \clip (0,0) rectangle (\width, \depth);
    
    \def\amplitude{0.2}
    \def\cycles{4} 
    
    \foreach \xStart in {0.2, 0.7, ..., \width} {
        \draw[chocolateLine, very thick] 
            plot[domain=0:\depth, samples=100] 
            ({\xStart + \amplitude * sin(\x * 360 * \cycles / \depth)}, \x);
    }
    
    \draw[thin, black!20] (0,0) rectangle (\width, \depth);
\end{scope}
\end{tikzpicture}
    \caption{The set of Fritz John points}
    \label{fig:millefeuille}
\end{figure}

\subsection{A system of polynomial inequalities}\label{subsec:systemofPolynomial}

To formally define a system of inequalities analogous to the one presented at the end of \Cref{ssec:triangular}, we first define some notation. 
Throughout this section, we assume, up to renaming them, that our dice form the set $\Dice = \{1, \dots, n\}$.
We let $k = |\Actions_\dev|$.

We first define the notion of \emph{strategy scheme}.
A $k$-grid strategy is such that the roll space can be cut into $k^n$ rectangles, such that on each rectangle, the tuple of actions that are prescribed to the team players is constant.
Thus, a $k$-grid strategy is defined by the thresholds at which we cut the roll space, and by the actions prescribed in each of those rectangles.
This second part is what a strategy scheme $\bbaction$ captures: for each indices $j_1, \dots, j_n$, and each player $p$, the action $a_{j_1 \dots j_n p}$ is the action played by player $p$ when the die $1$ falls into the $j_1$th interval, \dots, and the die $n$ falls into the $j_n$th interval.
A $\bbaction$-strategy is a $k$-grid strategy that is compatible with the strategy scheme $\bbaction$.

\begin{definition}[Strategy scheme, $\bbaction$-strategy]
    Let $\DGame$ be a dicey game.
    A \emph{strategy scheme} in $\DGame$ is a tuple of tuples:
    $$\bbaction = (a_{j_1 \dots j_n p})_{\substack{j_1, \dots, j_n \in [k]\\p \in \Pi}} \in \prod_{j_1, \dots, j_n \in [k]} \prod_{p \in \Pi} A_p.$$
    Given a strategy scheme $\bbaction$, we call \emph{$\bbaction$-strategy} a strategy such that there exists a tuple of tuples:
    $$\bblambda = (\lambda_{Dj})_{\substack{D \in \Dice \\ j \in [k]}} \in [0, 1]^{\Dice \times [k]}$$
    such that for each $D$, we have $\sum_j \lambda_{Dj} = 1$, and that for every roll $\broll$, if $j_1, \dots, j_n$ are the indices such that $\roll_D \in (\lambda_{Dj_1} + \dots + \lambda_{Dj_{D-1}}, \lambda_{Dj_1} + \dots + \lambda_{Dj_D}]$ for each die $D$, then the action profile $\bstrat(\broll)$ is equal to the action profile $\baction_{j_1 \dots j_n}$.
\end{definition}

Note that every $k$-grid strategy is an $\bbaction$-strategy for some strategy scheme $\bbaction$, and that conversely, every $\bbaction$-strategy is $k$-grid.
The following lemma follows from the definition of $\bbaction$-strategies.

\begin{lemma}\label{lm:inequations}
    Let $\bbaction$ be a strategy scheme.
    Then, for every real number $t$, there exists an $\bbaction$-strategy that guarantees at least the expected payoff $t$ if and only if there exists a solution of the form $(\bblambda, t)$ to the system of inequalities $\Sigma_{\bbaction}$, defined as follows:
    \begin{equation}\label{eq:positive}
        \forall D \in \Dice, \forall j \in [k], \lambda_{Dj} \geq 0
    \end{equation}
    \begin{equation}\label{eq:sum1}
        \forall D \in \Dice, \sum_{j=1}^k \lambda_{Dj} - 1 = 0
    \end{equation}
    \begin{equation}\label{eq:polyb}
        \forall b \in \Actions_\dev, \sum_{j_1=1}^k \dots \sum_{j_n=1}^k \mu(\baction_{j_1 \dots j_n}, b) \lambda_{1 j_1} \dots \lambda_{n j_n} - t \geq 0
    \end{equation}
\end{lemma}

Using \Cref{thm:itsOKtobestraight}, computing the optimal value a strategy can ensure can be done by, first, finding the optimal strategy scheme $\bbaction$, and second, finding a solution to the system $\Sigma_{\bbaction}$ that maximises the coordinate $t$.
With the example presented in \Cref{sec:example}, we have already shown that this optimal value may be an irrational number.
However, we will see in the next subsection that it is algebraic, with a finite representation whose size can be bounded.

\subsection{Optimal collective strategies}\label{subsec:optimalcollective}

\begin{theorem}\label{thm:exponentialbitsize}
    Let $\DGame$ be a dicey game.
    Then, there exists an optimal collective strategy in $\DGame$, which has a finite representation whose size is bounded by an exponential function of the size of $\DGame$.
    Moreover, the same bound applies to the value of that strategy.
\end{theorem}

\begin{proof}
    The full statement of the algebraic geometry theorems that are used in this proof can be found in App.~\ref{app:hammers}.

    First, let us observe that by \Cref{thm:itsOKtobestraight}, the quantity:
    $$\sup \{\val(\bstrat) \mid \bstrat \text{ is a collective strategy}\}$$
    equals the quantity
    $\max_{\bbaction} \sup\{\val(\bstrat) \mid \bstrat \text{ is a $\bbaction$-strategy}\}.$
    We can therefore choose a strategy scheme $\bbaction$ that realises the maximum above, and focus on $\bbaction$-strategies.
    Over the variables $\varLambda_{Dj}$, for $D \in \Dice$ and $j \in [k]$, and $T$, let us define the polynomials $P_i$, for each index $i \in I = \Actions_\dev \cup \Dice \times \{1, \dots, k, \geq, \leq\}$, as follows:
    \begin{itemize}
        \item if $i = (D, j) \in \Dice \times [k]$, then $P_i = \varLambda_{Dj}$.

        \item If $i = (D, \geq)$ for some die $D$, then $P_i = \sum_j \varLambda_{Dj} - 1$.

        \item If $i = (D, \leq)$ for some die $D$, then $P_i = 1 - \sum_j \varLambda_{Dj}$.
    
        \item If $i$ is an action $b \in \Actions_\dev$, then:
        $$P_b = \sum_{j_1} \dots \sum_{j_n} \mu(\bbaction, b) \varLambda_{1j_1} \dots \varLambda_{nj_n} - T.$$
    \end{itemize}

    Let $S \subseteq \Rb^{\Dice \times [k]} \times \Rb$ be the set of all pairs $(\bblambda, t)$ such that $P_i(\bblambda, t) \geq 0$ for each $i \in I$.
    A rephrasing of \Cref{lm:inequations} is that for every $t \in \Rb$, there exists an $\bbaction$-strategy of value at least $t$ if and only if there exists a pair $(\bblambda, t) \in S$.

    Over the set $S$, let us consider the projection $f: (\bblambda, t) \mapsto t$.
    By Weierstrass' extreme value theorem, as a continuous function over a compact set, it has a maximum---hence an optimal collective strategy exists.
    Let us now focus on the points where it finds that maximum.
    By Fritz John's necessary conditions~\cite[Theorem 4.2.8]{bazaraa2006nonlinear}, every such point is necessarily a \emph{Fritz John's point} of the function $f$, i.e., a point $(\bblambda, t)$ such that there exist coefficients $\alpha_0$ and $\alpha_i$ for each $i \in I$, all non-negative and not all zero, such that we have:
    $$\alpha_0 \nabla f(\bblambda, t) + \sum_{i \in I} \alpha_i \nabla P_i(\bblambda, t) = \bzero,$$
    with $\alpha_i = 0$ for every $i$ that is not an \emph{active constraint} at point $(\bblambda, t)$---the set $\Active(\bblambda, t) \subseteq I$ of active constraints is defined as the set of indices $i \in I$ such that $P_i(\bblambda, t) = 0$.    
    Let $F$ denote the set of all Fritz John points of $f$.
    We first state a technical proposition.
    
    \begin{restatable}[App.~\ref{app:proofOflinearalgebraProposition}]{proposition}{alphaneq}\label{prop:alpha0neq0}
        Let $(\bblambda, t) \in F$, and let $\alpha_0, (\alpha_i)_i$ be the corresponding coefficients.
        Then, we have $\alpha_0 \neq 0$.
    \end{restatable}
    \begin{proof}[Proof sketch]
        We assume toward contradiction that we have $\alpha_0=0$, and therefore:
        $$\sum_{i \in I} \alpha_i \nabla P_i(\bblambda, t) = \bzero.$$
        First, we show that for each $b\in \Actions_\dev$, we have $\alpha_b = \bzero$.
        Then, for all other active constraints, we show that the vectors $\nabla P_i(\bblambda, t) $ are linearly independent: then, all coefficients $\alpha_i$ are $0$, which is excluded by their definition.
    \end{proof}

We now turn our interest to the shape of the set $F$.

    \begin{proposition} \label{prop:FritzJohnfinite}
        The set $f(F)$ is finite.
    \end{proposition}

    \begin{proof}
        As a projection of a semialgebraic set, the set $F$ is semialgebraic.
        Then, by \cite[Theorem 5.38]{BPR06book}, there exists a finite partition $(S_\l)_{\l \in [m]}$ of $F$ such that each set $S_\l$ (called \emph{stratum}):
        \begin{itemize}
            \item is connected;
            \item and is such that the set of active constraints $\Active(\bblambda, t)$ is the same for all points $(\bblambda, t) \in S_\l$---we denote it by $I_\l$.
        \end{itemize}

        Our result will be proven if we show that the function $f$ is constant on each stratum.
        Let us therefore fix $\l \in [m]$.
        Let $\gamma: [0, 1] \to S_\l$ be a smooth curve on $S_\l$.
        We prove that the mapping $f \circ \gamma: [0, 1] \to \Rb$ is constant; since the set $S_\l$ is connected, that will be enough to end our proof.
        To do so, we prove that the derivative of this function is everywhere zero.
        Let $s \in [0, 1]$.
        We have:
        $$(f \circ \gamma)'(s) = \left\langle \left. \nabla f(\gamma(s)) ~\right|~ \gamma'(s) \right\rangle,$$
        where $\langle \cdot | \cdot \rangle$ denotes the canonical scalar product over the space $\Rb^{\Dice \times [k]} \times \Rb$.
        Now, for each such $s$, we can define the corresponding coefficients $\alpha_0$ and $\alpha_i$ for $i \in I$, and since we have $\alpha_0 \neq 0$ by \Cref{prop:alpha0neq0}, we can write:
        $$f(\gamma(s)) = -\sum_{i \in I} \frac{\alpha_i}{\alpha_0} P_i,$$
        which, inserted into the previous equality, yields:
        $$(f \circ \gamma)'(s) = \left\langle \left. \nabla \left( -\sum_{i \in I} \frac{\alpha_i}{\alpha_0} P_i\right)(\gamma(s)) ~\right|~ \gamma'(s) \right\rangle$$
        $$= -\sum_{i \in I} \frac{\alpha_i}{\alpha_0} \left\langle \left. \nabla P_i(\gamma(s)) ~\right|~ \gamma'(s) \right\rangle.$$
        To prove that this sum equals $0$, we now show that each term actually equals $0$.
        To do so, we separate the two cases:
        \begin{itemize}
            \item for $i \not\in I_\l$, the constraint $i$ is inactive for all points of the stratum $S_\l$, hence $\alpha_i = 0$.

            \item For $i \in I_\l$, by definition of active constraints, we have $P_i(\gamma(s')) = 0$ for all $s'$.
            By taking the derivatives, we get:
            $$\left\langle \left. \nabla P_i(\gamma(s)) ~\right|~ \gamma'(s) \right\rangle = 0.$$
        \end{itemize}

        The whole sum above is thus zero, i.e., we have $(f \circ \gamma)'(s) = 0$ for all $s \in [0, 1]$, hence the function $f \circ \gamma$ is constant over $[0, 1]$, and therefore the function $f$ is constant over the stratum $S_\l$.
        Since there are finitely many strata, the result follows.
    \end{proof}

    We can now use this result to conclude.
    As already observed, the set $F$ is a semialgebraic set.
    We can therefore apply a result by Basu, Pollack and Roy~\cite[Theorem~13.11]{BPR06book}, which says that in such a set, each connected component contains a point whose coordinates are all algebraic, and whose size is $LD^{O(N)}$, where:
    \begin{itemize}
        \item $L$ is the maximal size of the coefficients of all polynomials $P_i$, which in our case is bounded by the size of $\DGame$.
        \item $D$ is the maximal degree of those polynomials, in our case~$n$.
        \item $N$ is the dimension of the space, in our case $nk+1$.
    \end{itemize}

    By \Cref{prop:FritzJohnfinite}, the projection $f(F)$ is finite.
    There is therefore a connected component $K$ of $F$ such that $f(K) = \{\max f\}$; and thus, from that connected component, we can pick a point $(\bblambda, t) \in K \subseteq S$ such that $t = \max f$, and whose size is $LD^{O(N)}$, i.e., is bounded by an exponential function of the size of the dicey game $\DGame$.
    This point defines an optimal strategy in $\DGame$ and its value.
\end{proof}

\section{Algorithms and complexities}\label{sec:algos}
As tradition all but requires, we round up our structural results with a few computational complexity bounds on problems related to our newly introduced model.
\emph{Threshold} problems, i.e., deciding whether team players can ensure a certain value, can be done by solving the system of inequations defined \Cref{lm:inequations}, a practical consequence of \Cref{thm:itsOKtobestraight}.
The \emph{value computation} problems---computing the $\max\min$ value of a dicey game---are more requiring, and our upper bounds also rely on the second structural result provided by \Cref{thm:exponentialbitsize}.
The lower bounds are proven with reductions from classical problems.

\subsection{The threshold problem}

The first problem we study is a decision problem: in a given dicey game, is it possible to achieve an expected payoff at least as high as a specified threshold?

\begin{problem}[Threshold problem]
    Given a dicey game $\DGame$ and a threshold $\threshold \in \Qb$, is there a collective strategy $\bstrat$ in $\DGame$ with $\val(\bstrat) \geq \threshold$?
\end{problem}

Using \Cref{lm:inequations}, this problem can be reduced to solving a system of inequations of exponential size, hence the following.

\begin{lemma} \label{lm:threshold_expspace_easy}
    The threshold problem is in $\EXPSPACE$\footnote{Our proof actually locates this problem in the class $\textsf{succ}\exists\Rb$, as defined in the work of Blaser et al.~\cite{BDJMvdZ24}.
    This class is known to be included in $\EXPSPACE$.
    It is however unlikely that our problem is complete for this class, since we reduce it to solving a system of inequations that has exponential size, but polynomially many variables---which has been used in the proof of \Cref{thm:exponentialbitsize} to prove that it has an optimal solution of exponential size, which would not be true for every instance of the \emph{succinct existential theory of the reals}, the canonical complete problem for the class $\textsf{succ}\exists\Rb$.}.
\end{lemma}

\begin{proof}
    Given a dicey game $\DGame$ and a threshold $t$, deciding whether there exists a collective strategy in $\DGame$ with value at least $t$ can be done by, first, guessing a strategy scheme $\bbaction$ (which can be done in exponential time), and then checking whether the system defined by \Cref{lm:inequations} has a solution.
    That system can be written as a formula in the existential theory of the reals, and such formulas can be decided in $\PSPACE$.
    Since the system's size is exponential in the size of the dicey game $\DGame$, we obtain membership to the class $\NEXPSPACE$, which is equal to $\EXPSPACE$ by Savitch's theorem.
\end{proof}
   
We now give a lower bound.

\begin{lemma} \label{lm:threshold_nexptime_hard}
    The threshold problem is $\NEXP$-hard.
\end{lemma}

\begin{proof}
    We proceed by reduction from the $\NEXP$-hard problem $\DQBF$: Given a formula of the form:
    $$\forall x_1, \dots, \forall x_n, \exists y_1(X_1), \dots, \exists y_m(X_m), \bigwedge_{k=1}^q C_k,$$
    where each $X_j$ is a subset of $\{x_1, \dots, x_n\}$ and each $C_k$ is a clause of three literals over the variables $x_1, \dots, x_n, y_1, \dots, y_m$, does there exist a collection of mappings $f_j: \{0, 1\}^{X_j} \to \{0, 1\}$ such that for every valuation $\nu: \{x_1, \dots, x_n\} \to \{0, 1\}$, the valuation:
    $$\begin{cases}
        x_i & \mapsto \nu(x_i) \\
        y_j & \mapsto f_j\left(\nu_{|X_j}\right)
    \end{cases}$$
    satisfies all the clauses $C_k$?

    \paragraph*{Definition of the dicey game $\DGame$}

    Let us consider such a formula $\phi$.
    From $\phi$, we define the dicey game $\DGame$ as follows.
    The team is composed of $n+m$ players who are  named after the variables $x_1, \dots, x_n, y_1, \dots, y_m$.
    For each universal variable $x_i$, there is a die $\Die_i$, accessible by player $x_i$, and by every player $y_j$ such that $x_i \in X_j$.
    Each player other than the Devil chooses between two actions: $0$ and $1$, and the Devil chooses between the actions $\rand^0_1, \dots, \rand^0_n, \rand^1_1,$ $\dots,$ $\rand^1_n, C_1, \dots$, and~$C_q$.

    When the Devil chooses the action $\rand_i^b$, with $i \in [n]$ and $b \in \{0, 1\}$, the payoff is $2$ if player $x_i$ picks the action $b$, and $0$ otherwise.
    When the Devil chooses the action $C_k$, the payoff is $1$ if there is a positive literal $z$ in $C_k$ such that player $z$ plays the action $1$, or if there is a negative literal $\neg z$ in $C_k$ such that player $z$ plays the action $0$; otherwise, the payoff is $0$.
    The threshold is $\threshold = 1$.

    \paragraph*{If the formula is true, then the team has a collective strategy to get at least the expected payoff $1$}

    Let us assume that the formula is true, and that the mappings $f_1, \dots, f_m$ exist.
    
    Let us define a collective strategy $\bstrat$.
    For each $i$, player $x_i$ plays $0$ if the outcome of the die $D_i$ is below $0.5$, and $1$ otherwise.
    Thus, if the Devil chooses an action of the form $\rand_i^b$, the payoff $2$ is obtained with probability $0.5$, and the expected payoff is $1$.
    For each $j$, player $y_j$ has access to all dice $D_i$ such that $x_i \in X_j$, and therefore knows what every player $x_i \in X_j$ does; he can then play the action $f_j((\action_i)_{i \in X_j})$, where each $\action_i$ is the action played by player $x_i$.
    Thus, when the Devil picks an action of the form $C_k$, whatever actions the players $x_i$ choose, the players $y_j$ choose their action so that the corresponding valuation satisfies all clauses, and in particular the clause $C_k$: the payoff obtained is therefore always $1$.

    Whatever the Devil does, the expected payoff is therefore always $1$, which is then the value of the collective strategy $\bstrat$.

    \paragraph*{If the team has a collective strategy with value at least $1$, then the formula is true} Let $\bstrat$ be a collective strategy with value at least $1$.
    
    Using \Cref{thm:itsOKtobestraight}, we can assume that $\bstrat$ is $(2n+q)$-grid.
    Let us consider a player $x_i$.
    If there was an action $b \in \{0, 1\}$ that player $x_i$ would choose with probability greater than $0.5$, then the Devil could deterministically choose the action $\rand_i^{1-b}$, so that the expected payoff drops below $1$.
    This would be in contradiction with the hypothesis $\val(\bstrat) \geq 1$: therefore, player $x_i$'s strategy necessarily uses the only die he has access to, the die $\Die_i$, to play both actions with probability exactly $0.5$.
    Since the collective strategy $\bstrat$ is $(2n+q)$-grid, we can pick a non-empty interval $[u_{0i}, v_{0i}]$ on which all rolls are equivalent from all players' perspectives, and such that for every roll $r \in [u_{0i}, v_{0i}]$, we have $\strat_{x_i}(r) = 0$; and analogously, we can choose a similar interval $[u_{1i}, v_{1i}]$ such that for every roll $r \in [u_{1i}, v_{1i}]$, we have $\strat_{x_i}(r) = 1$.

    Let now $j \in [m]$.
    We define the mapping $f_j$ as follows: for every valuation $\nu: X_j \to \{0, 1\}$, we define $f_j(\nu)$ as the action chosen by player $y_j$ when for each variable $x_i \in X_i$, the outcome of the die $\Die_i$ falls in the interval $\left[u_{\nu(x_i)i}, v_{\nu(x_i)i}\right]$.

    Let then $\nu: \{x_1, \dots, x_n\} \to \{0, 1\}$ be a valuation.
    Let $C_k$ be a clause.
    Since we have $\val(\bstrat) \geq 1$, if the Devil plays the action $C_k$, the expected payoff is necessarily at least $1$; which means that the payoff is almost surely $1$.
    In particular, there is a positive probability that the outcome of each die $\Die$ falls in the interval $\left[u_{\nu(x_i)i}, v_{\nu(x_i)i}\right]$: in such a case, the payoff obtained by the team is necessarily $1$, which means that the actions played by the team players form a valuation that satisfies the clause $C_k$.
    In this event, the valuation formed by the team players' actions is exactly the valuation $\nu$ extended by the mappings $f_j$; which proves the validity of the formula $\phi$.
\end{proof}

We can thus conclude to the following theorem.

\begin{theorem} \label{thm:threshold_complexity}
    The threshold problem is in the complexity class $\EXPSPACE$ and is also $\NEXP$-hard.
\end{theorem}

\subsection{Value computation}

We now consider the following functional problem.

\begin{problem}[Value computation]
    Given a dicey game $\DGame$, what is the maximal value a collective strategy in $\DGame$ can guarantee?
\end{problem}

This problem lies in the functional version of $\EXPSPACE$.
We do not provide a lower bound, as it is already implied by \Cref{thm:threshold_complexity}: once the value of a dicey game is known, deciding the threshold problem is immediate, which implies that the value computation problem is hard for the class $\NEXP$.

\begin{theorem} \label{thm:valuecomputation}
    There is an algorithm that solves the value computation problem in exponential space.
\end{theorem}

\begin{proof}
    Let $\DGame$ be a dicey game, and let $s$ be its size.
    By \Cref{thm:exponentialbitsize}, there exists a exponential function $e: \Nb \to \Nb$ such that the best value $t^\star$ that a collective strategy can achieve in this game has size smaller than $e(s)$.
    Our algorithm enumerates all possible values $t$ of size less that $e(s)$, and checking, for each of them, whether there exists a collective strategy that achieves it; and remembering the largest such $t$.
    This can be done in exponential space, using \Cref{thm:threshold_complexity}.
\end{proof}

\subsection{Bounded number of accessible dice}

The high complexities found in the previous subsections are the consequences of an exponential blow-up: by dividing the set of possible rolls for each die into $k$ intervals, the total set of possible rolls for a player that has access to $d$ dice is divided into $k^d$ rectangular sets, for which an action must be chosen.
In practice, however, it is reasonable to believe that each player has access to a small number of sources of randomness; all the more so as, as we will see in \Cref{sec:moreexamples}, that does not seem to be a strong disadvantage.
We therefore investigate the case where the number of dice each player has access to is bounded by a fixed integer: the \emph{die bound}.

The upper bound is, again, a consequence of \Cref{lm:inequations}.

\begin{lemma} \label{lm:threshold_existsR_easy}
    The threshold problem with fixed die bound is in $\NP^{\exists\Rb}$.
\end{lemma}

\begin{proof}
    By \Cref{thm:itsOKtobestraight,lm:inequations}, an algorithm that solves the threshold problem is the following: guess a strategy scheme $\bbaction$, then check that the system $\Sigma_\bbaction$ (as defined by \Cref{lm:inequations}) has a solution.
    With fixed die bound, strategy schemes have polynomial size, and so does the system $\Sigma_\bbaction$, hence the desired complexity.
\end{proof}

To show a lower bound, we need the following lemma.

\begin{restatable}[App.~\ref{app:intermediary_hardness_result}]{lemma}{ETRhardInitialProblem}\label{lemma:ETRhardInitialProblem}
    The following problem is $\exists\Rb$-hard: given a system $S$ of inequations $I$ of the form $I: \sum_{i=1}^n \sum_{j=1}^n \lambda_{ij} x_i x_j \geq 0$, with $\lambda_{ij} \in \Zb$ for each $i, j$ and variables $x_1, \dots, x_n$, does the system $S$ have a \emph{normal} solution, i.e. a solution satisfying all equations in $S$ as well as the inequality $x_i \geq 0$ for each $i$, and the equality $\sum_i x_i = 1$?
\end{restatable}

Using that result, we show the following.

\begin{lemma} \label{lm:threshold_existsR_hard}
    The threshold problem with fixed die bound is $\exists\Rb$-hard, even with only two team players, each having access to one individual die (and, therefore, even for $d = 1$).
\end{lemma}

\begin{proof}
    We reduce from the problem defined by \Cref{lemma:ETRhardInitialProblem}.

    \paragraph*{Definition of the dicey game $\DGame$}
    Let $S$ be a set of inequations of the form given in \Cref{lemma:ETRhardInitialProblem}.
    Let us construct a dicey game $\DGame$ with three players, where the team players are called \emph{Ada} and \emph{Bertrand}.
    We will make sure that Ada and Bertrand have strategies to get an expected payoff greater than or equal to $0$ if and only if the system $S$ has a normal solution.
    
    Ada's actions are variables $x_1$, \dots, $x_n$, and similarly for
    Bertrand.
    Intuitively, the probability which which both Ada and Bertrand play each action $x_i$ will correspond to the value given to the variable $x_i$.
    
    The Devil's actions are $x_1^\leq$, \dots, $x_n^\leq$, $x_1^\geq$, \dots, $x_n^\geq$, and all inequations $I \in S$.
    Each of those actions will be used to impose a constraint on Ada's and Bertrand's strategies: each action $x_i^\leq$ will impose that Ada's probability of playing the action $x_i$ is smaller than or equal to Bertrand's one; each action $x_i^\geq$ will impose that it is greater than or equal to it; and the action $I$ will impose that those probabilities satisfy the inequation $I$.

    Let us now define payoffs.
    \begin{itemize}    
        \item If the Devil plays the action $x_i^\leq$:
        \begin{itemize}
            \item if Ada and Bertrand both play $x_i$, they get payoff $0$.
            \item If Ada plays $x_i$ and Bertrand plays anything else, they get payoff $-1$.
            \item If Ada plays anything else and Bertrand play $x_i$, they get payoff $1$.
            \item If neither of them plays $x_i$, they get payoff $0$.
        \end{itemize}

        \item If the Devil plays the action $x_i^\geq$, the payoffs are the opposite of the payoffs given above.

        \item If the Devil plays the action $I: \sum_{i=1}^n \sum_{j=1}^n \lambda^I_{ij} x_i x_j \geq 0$, Ada plays the action $x_i$, and Bertrand plays the action $x_j$, then Ada and Bertrand get payoff $\lambda^I_{ij}$.
    \end{itemize}

    Finally, Ada has access to one die, and Bertrand to another die.

    \paragraph*{If the system $S$ has a normal solution, then Ada and Bertrand have a collective strategy to get a non-negative expected payoff.}

        Let $\hx_1, \dots, \hx_n \in [0, 1]$ constitute a normal solution of the system $S$, and let $\strat_\ada$ and $\strat_\bert$ be strategies for Ada and Bertrand that both consist of playing each action $x_i$ with probability $\hx_i$ (which is possible since their sum equals $1$).
        Then, we need to check that for every action the Devil can choose, the expected payoff is non-negative.

        If the Devil plays the action $x_i^\leq$, then the expected payoff is:
        $$\hx_i \hx_i \times 0 + \hx_i (1 - \hx_i) (-1) + (1 - \hx_i) \hx_i \times 1 + (1 - \hx_i)(1 - \hx_i)  \times 0,$$
        i.e. $0$.
        And symmetrically, it is $0$ if the Devil plays the action $x_i^\geq$.
        If the Devil plays the action $I$, then the expected payoff is the sum
        $\sum_i \sum_j \hx_i \hx_j \lambda_{ij}^I,$
        which is also non-negative by hypothesis.
        In all cases, the team secures a non-negative expected payoff.

\paragraph*{If Ada and Bertrand have a collective strategy to get a non-negative expected payoff, then the system $S$ has a normal solution.}

    Let $\strat_\ada$ and $\strat_\bert$ be such strategies.
    Let us write $\hx_1, \dots, \hx_n$ for the probabilities with which Ada chooses the actions $x_1, \dots, x_n$, respectively, and let $\hy_1, \dots, \hy_n$ be the same probabilities for Bertrand.

    Let us first prove that we have $\hx_i = \hy_i$ for each $i$.
    If the Devil plays the action $x_i^\leq$, then the expected payoff obtained is:
    $$\hx_i \hy_i \times 0 + \hx_i (1 - \hy_i) (-1) + (1 - \hx_i) \hy_i \times 1 + (1 - \hx_i)(1 - \hy_i) \times 0,$$
    i.e. $-\hx_i + \hy_i$.
    Symmetrically, if the Devil plays the action $x_i^\geq$, then the expected payoff obtained is $\hx_i - \hy_i$.
    By hypothesis, those two quantities are non-negative, which implies $\hx_i = \hy_i$.

    Now, let us show that the reals $\hx_1, \dots, \hx_n$ constitute a simple solution of the system $S$.
    Let $I \in S$.
    If the Devil plays the action $I$, then the expected payoff obtained is  $\sum_i \sum_j \hx_i \hy_j \lambda^I_{ij}$.
    Since this sum is non-negative, the inequation $I$ is satisfied.
    The values $\hx_1, \dots, \hx_n$ constitute then a solution of the system $S$, which is simple by construction.
    Since the dicey game $\DGame$ can be constructed in polynomial time from $S$, this reduction proves the desired hardness result.
\end{proof}

Let us gather those results in one theorem.

\begin{theorem}
    If there is a fixed number $d \in \Nb$ such that no player has access to more than $d$ dice, then the threshold problem is in $\NP^{\exists\Rb}$ and $\exists\Rb$-hard.
\end{theorem}

    \subsection{Allocating dice to players}

We now turn our attention to the case where the dice structure is not a part of the input to an instance of the problem, and define the corresponding variants of our problems.

\begin{problem}[Allocate threshold problem]
    Given a game $\Game$, a dice pack $\Pack$, and a threshold $\threshold \in \Qb$, does there exist a dice structure $\Struct$ matching the dice pack $\Pack$ and a collective strategy in the dicey game $(\Game, \Struct)$ whose value is greater than or equal to $\threshold$?
\end{problem}

\begin{problem}[Allocate value computation problem]
    Given a game $\Game$ and a dice pack $\Pack$, what is the maximal value a collective strategy can guarantee, in a dicey game $(\Game, \Struct)$, for some dice structure $\Struct$ that matches the dice pack $\Pack$?
\end{problem}

Perhaps surprisingly, these problems fall in the same complexity classes as their previous versions.

\begin{lemma} \label{lm:allocate_value_computation}
    There is an algorithm that solves the allocate value computation problem using exponential space.
\end{lemma}

\begin{proof}
    By \Cref{thm:valuecomputation}, when the dice structure is given, the optimal value a collective strategy can achieve can be computed in exponential space.
    Moreover, all dice structures matching a given dice pack can be represented with a polynomial number of bits.
    Thus, the allocate value computation problem can be solved by enumerating all possible dice structures, calling that algorithm for each of them, and remembering the highest value found.
\end{proof}

This yields an algorithm for the allocate threshold problem.

\begin{lemma}
    The allocate threshold problem is in $\EXPSPACE$.
\end{lemma}

\begin{proof}
    By \Cref{lm:allocate_value_computation}, given a game and a dice pack, there is an algorithm to compute the best possible value in exponential space.
    By comparing that value to the threshold $\threshold$, that algorithm can also be used to decide the allocate threshold problem.
\end{proof}

For that problem, our previous hardness result is also still valid.

\begin{restatable}[App.~\ref{app:allocate_threshold_hard}]{lemma}{allocateThresholdHard} \label{lm:allocate_threshold_hard}
    The allocate threshold problem is $\NEXP$-hard.
\end{restatable}

Let us now turn our attention to the fixed dice bound case.

\begin{lemma}
    The allocate threshold problem with fixed die bound is in $\NP^{\exists\Rb}$.
\end{lemma}

\begin{proof}
    Our algorithm for the threshold problem with fixed die bound consisted in guessing a strategy scheme $\bbaction$, and then call the $\exists\Rb$ oracle to check whether there exists an $\bbaction$-strategy guaranteeing the threshold.
    The same algorithm works here, if we also guess a dice structure together with the strategy scheme.
\end{proof}

Again, we also have the same hardness result as before.

\begin{lemma}
    The allocate threshold problem with fixed die bound is $\exists\Rb$-hard.
\end{lemma}

\begin{proof}
    \Cref{lm:threshold_existsR_hard} shows the the threshold problem with fixed die bound is $\exists\Rb$-hard even on dicey games where each player has access to exactly one individual die.
    That problem reduces to the allocate threshold problem with fixed die bound, even with $d=1$, as follows.
    Given a dicey game $\DGame = (\Game, \Struct)$ and a threshold $t$, we return the game $\Game$, the dice pack $\Pack = (\Dice, \acc)$ with $\acc: \Die \mapsto 1$, and the threshold $t$.
    Then, there exists a collective strategy in $\DGame$ with value at least $t$ if and only if there exists a dice structure $\Struct'$ matching $\Pack$ and a collective strategy in the dicey game $(\Game, \Struct')$ with value at least $t$.
    Which shows $\exists\Rb$-hardness of our problem.
\end{proof}

We can then conclude this section with the following theorem.

\begin{theorem} \label{thm:allocate_complexity}
    There is an algorithm that solves the allocate value computation problem using exponential space.
    As for the allocate threshold problem, it is in the class $\EXPSPACE$ and $\NEXP$-hard.
    With a fixed die bound $d$, it is in the class $\NP^{\exists\Rb}$ and is $\exists\Rb$-hard.
\end{theorem}

\section{Interesting examples} \label{sec:moreexamples}
    We 
    present some results on small problem instances which the reader might find surprising, leading up to a conjecture.

    \begin{figure}[h]
    \newcommand{\s}{0.6}
    \centering
    \begin{subfigure}[t]{0.5\textwidth}
    \centering
	\begin{tikzpicture}[scale=\s]  
        \draw (0,0) node[below left] {\scriptsize Player 1};
    	\draw (0,1) node[above left] {\scriptsize Player 2};
    	\draw (1,0) node[below right] {\scriptsize Player 3};
    	\draw (1,1) node[above right] {\scriptsize Player 4};
        \draw (0,0) node {$\bullet$};
    	\draw (0,1) node {$\bullet$};
    	\draw (1,0) node {$\bullet$};
    	\draw (1,1) node {$\bullet$};

        \draw (1,1) -- (0,0) -- (0,1) -- (1,1) -- (1,0) -- (0,0);
        \draw (0, 1) -- (1, 0);
    \end{tikzpicture}
    \caption{$4$-clique matching pennies}
    \label{fig:4clique}
    \end{subfigure}
    \hspace{1mm}
    \begin{subfigure}[t]{0.4\textwidth}
    \centering
	\begin{tikzpicture}[scale=\s]  
        \draw (0,0) -- (0,1);
        \draw (1,0) -- (1,1);
        \draw[gray!50 ,dashed] (0,0) -- (1,0) -- (0,1) -- (1,1) -- (0,0);
        \draw (0,0) node[below left] {\scriptsize Player 1};
    	\draw (0,1) node[above left] {\scriptsize Player 2};
    	\draw (1,0) node[below right] {\scriptsize Player 3};
    	\draw (1,1) node[above right] {\scriptsize Player 4};
        \draw (0,0) node {$\bullet$};
    	\draw (0,1) node {$\bullet$};
    	\draw (1,0) node {$\bullet$};
    	\draw (1,1) node {$\bullet$};
    \end{tikzpicture}
    \caption{An optimal strategy}
    \label{fig:4cliqueoptimal}
    \end{subfigure}
    \begin{subfigure}[t]{0.5\textwidth}
    \centering
    \begin{tikzpicture}[scale=\s]
        \node[circle, fill, inner sep=1.5pt] (P1) at (90:1) {};
        \node[circle, fill, inner sep=1.5pt] (P2) at (18:1) {};
        \node[circle, fill, inner sep=1.5pt] (P3) at (-54:1) {};
        \node[circle, fill, inner sep=1.5pt] (P4) at (-126:1) {};
        \node[circle, fill, inner sep=1.5pt] (P5) at (162:1) {};
    
        \node[above]       at (P1) {\scriptsize Player 1};
        \node[right]       at (P2) {\scriptsize Player 2};
        \node[below right] at (P3) {\scriptsize Player 3};
        \node[below left]  at (P4) {\scriptsize Player 4};
        \node[left]        at (P5) {\scriptsize Player 5};
    
        \foreach \u/\v in {
            P1/P2, P1/P3, P1/P4, P1/P5,
            P2/P3, P2/P4, P2/P5,
            P3/P4, P3/P5,
            P4/P5}
        {
            \draw (\u) -- (\v);
        }
    \end{tikzpicture}
    \caption{Another game with the Devil}
    \label{fig:5clique}
    \end{subfigure}
    \hspace{1mm}
    \begin{subfigure}[t]{0.4\textwidth}
    \centering
    \begin{tikzpicture}[scale=\s]
        \node[circle, fill, inner sep=1.5pt] (P1) at (90:1) {};
        \node[circle, fill, inner sep=1.5pt] (P2) at (18:1) {};
        \node[circle, fill, inner sep=1.5pt] (P3) at (-54:1) {};
        \node[circle, fill, inner sep=1.5pt] (P4) at (-126:1) {};
        \node[circle, fill, inner sep=1.5pt] (P5) at (162:1) {};
    
        \node[above]       at (P1) {\scriptsize Player 1};
        \node[right]       at (P2) {\scriptsize Player 2};
        \node[below right] at (P3) {\scriptsize Player 3};
        \node[below left]  at (P4) {\scriptsize Player 4};
        \node[left]        at (P5) {\scriptsize Player 5};
    
        \foreach \u/\v in {
            P1/P2, P1/P3, P2/P3, P4/P5}
        {
            \draw (\u) -- (\v);
        }
        \foreach \u/\v in {
            P1/P4, P1/P5,
             P2/P4, P2/P5,
            P3/P4, P3/P5}
        {\draw[gray!50,dashed] (\u) -- (\v);}
    \end{tikzpicture}
    \caption{An optimal strategy}
    \label{fig:5cliqueoptimal}
    \end{subfigure}
	\caption{A few more examples} \label{fig:moreexamples}
\end{figure}

    \subparagraph*{$4$-clique matching pennies.}
    
    This game is a new variant of the matching pennies game.
    Now, four players are playing against the Devil.
    Again, they all pick an action between Heads and Tails, and the team wins if all actions match.
    Each pair of team players share a die---and therefore, each player has access to three dice.
    That situation is illustrated by \Cref{fig:4clique}, where an edge between two players means that they share a die.

    Here, one might expect that the optimal strategy is similar to the one that was found for triangular matching pennies: one threshold $\lambda$ obtained by some intricate polynomial equation, and each player choosing Heads if all the rolls they see are above that threshold, and Tails otherwise.
    The reality is quite different.
    The $\max\min$ value of this game is $1/4$, which is obtained by the following strategy, illustrated by \Cref{fig:4cliqueoptimal}: players $1$ and $2$ use their common die to make sure they always match, playing both Heads or both Tails with probability $1/2$ each; and players $3$ and $4$ act similarly with their common die.
    The four other dice are not used at all.

    \subparagraph*{$5$-clique matching pennies.}
    We continue with five team players---again, each pair shares a common die, as in \Cref{fig:5clique}.
    Here, the $\max\min$ value is approximately 0.139.
    An optimal collective strategy is the following (see \Cref{fig:5cliqueoptimal}): players $4$ and $5$ use their common die to play both Heads or both Tails with probability 1/2 each.
    Meanwhile, players 1, 2, and 3 use their three common dice to play the strategy studied in \Cref{sec:example}.
    Again, all other dice are ignored.

    \subparagraph*{A conjecture.}
    These examples suggest that while shared sources of randomness can be powerful, only a small number of them may suffice to provide optimal solutions. Guided by this intuition, we propose the following conjecture, which we have verified 
    for all values of $n\leq 7$.

    \begin{conjecture}
    Consider a team of $n$ players playing matching pennies against the Devil, where every pair of players shares a die.
    The following strategy is optimal.
    
    If the number $n$ is even:
    players are grouped by pairs, and each pair uses its common die to play both Heads or both Tails with probability $1/2$ each.
        All other dice are ignored, and the value of this collective strategy is:
        $$\left(\frac{1}{2}\right)^\frac{n}{2}.$$

        If the number $n$ is odd: three players are grouped together, and use the three dice they share to follow the strategy described in \Cref{sec:example}.
        All other players are grouped by pairs as in the even case, and all other dice are ignored.
        The value of this collective strategy is:
        $$\left(\frac{1}{2}\right)^{\frac{n-3}{2}}\beta,$$ where $\beta \approx 0.2781$ is the number defined in \Cref{sec:example}.
    \end{conjecture}
    

\bibliography{biblio.bib}

\begin{thebibliography}{10}

\bibitem{AGHHKNPSW21}
Alessandro Abate, Julian Gutierrez, Lewis Hammond, Paul Harrenstein, Marta Kwiatkowska, Muhammad Najib, Giuseppe Perelli, Thomas Steeples, and Michael~J. Wooldridge.
\newblock Rational verification: game-theoretic verification of multi-agent systems.
\newblock {\em Appl. Intell.}, 51(9):6569--6584, 2021.
\newblock \href {https://doi.org/10.1007/S10489-021-02658-Y} {\path{doi:10.1007/S10489-021-02658-Y}}.

\bibitem{alonso2016persuading}
Ricardo Alonso and Odilon C{\^a}mara.
\newblock Persuading voters.
\newblock {\em American Economic Review}, 106(11):3590--3605, 2016.

\bibitem{AD74}
R.~J. Aumann and J.~H. Dreze.
\newblock Cooperative games with coalition structures.
\newblock {\em Int. J. Game Theory}, 3(4):217–237, December 1974.
\newblock \href {https://doi.org/10.1007/BF01766876} {\path{doi:10.1007/BF01766876}}.

\bibitem{Aum1974}
Robert~J. Aumann.
\newblock Subjectivity and correlation in randomized strategies.
\newblock {\em Journal of Mathematical Economics}, 1(1):67--96, 1974.
\newblock URL: \url{https://www.sciencedirect.com/science/article/pii/0304406874900378}, \href {https://doi.org/10.1016/0304-4068(74)90037-8} {\path{doi:10.1016/0304-4068(74)90037-8}}.

\bibitem{Aum74b}
Robert~J Aumann.
\newblock Subjectivity and correlation in randomized strategies.
\newblock {\em Journal of mathematical economics}, 1(1):67--96, 1974.

\bibitem{BPR94}
Saugata Basu, Richard Pollack, and Marie{-}Fran{\c{c}}oise Roy.
\newblock On the combinatorial and algebraic complexity of quantifier elimination.
\newblock In {\em {FOCS}}, pages 632--641. {IEEE} Computer Society, 1994.
\newblock \href {https://doi.org/10.1109/SFCS.1994.365728} {\path{doi:10.1109/SFCS.1994.365728}}.

\bibitem{BPR96}
Saugata Basu, Richard Pollack, and Marie-Fran{\c{c}}oise Roy.
\newblock On the combinatorial and algebraic complexity of quantifier elimination.
\newblock {\em Journal of the ACM (JACM)}, 43(6):1002--1045, 1996.

\bibitem{BPR06book}
Saugata Basu, Richard Pollack, and Marie-Fran{\c{c}}oise Roy.
\newblock {\em Algorithms in Real Algebraic Geometry}.
\newblock Springer, 2nd edition, 2006.
\newblock \href {https://doi.org/10.1007/978-3-540-30140-0} {\path{doi:10.1007/978-3-540-30140-0}}.

\bibitem{BPR97}
Saugata Basu, Richard Pollack, and Marie-Françoise Roy.
\newblock On computing a set of points meeting every cell defined by a family of polynomials on a variety.
\newblock {\em Journal of Complexity}, 13(1):28--37, 1997.
\newblock \href {https://doi.org/10.1006/jcom.1997.0434} {\path{doi:10.1006/jcom.1997.0434}}.

\bibitem{bazaraa2006nonlinear}
Mokhtar~S. Bazaraa, Hanif~D. Sherali, and C.~M. Shetty.
\newblock {\em Nonlinear Programming: Theory and Algorithms}.
\newblock John Wiley \& Sons, Hoboken, New Jersey, 3rd edition, 2006.

\bibitem{Bellare2020}
Mihir Bellare, Wei Dai, and Phillip Rogaway.
\newblock Reimagining secret sharing: {Creating} a safer and more versatile primitive by adding authenticity, correcting errors, and reducing randomness requirements.
\newblock In {\em Proceedings on Privacy Enhancing Technologies (PoPETs)}, volume 2020, pages 461--490, 2020.
\newblock \href {https://doi.org/10.2478/popets-2020-0081} {\path{doi:10.2478/popets-2020-0081}}.

\bibitem{bergemann2016bayes}
Dirk Bergemann and Stephen Morris.
\newblock Bayes correlated equilibrium and the comparison of information structures in games.
\newblock {\em Theoretical Economics}, 11(2):487--522, 2016.

\bibitem{bergemann2019information}
Dirk Bergemann and Stephen Morris.
\newblock Information design: A unified perspective.
\newblock {\em Journal of Economic Literature}, 57(1):44--95, 2019.

\bibitem{BDJMvdZ24}
Markus Bl\"{a}ser, Julian D\"{o}rfler, Maciej Li\'{s}kiewicz, and Benito van~der Zander.
\newblock {The Existential Theory of the Reals with Summation Operators}.
\newblock In {\em ISAAC 2024}, volume 322 of {\em LIPIcs}, pages 13:1--13:19. Schloss Dagstuhl -- Leibniz-Zentrum f{\"u}r Informatik, 2024.
\newblock \href {https://doi.org/10.4230/LIPIcs.ISAAC.2024.13} {\path{doi:10.4230/LIPIcs.ISAAC.2024.13}}.

\bibitem{brandl2016consistent}
Florian Brandl, Felix Brandt, and Hans~Georg Seedig.
\newblock Consistent probabilistic social choice.
\newblock {\em Econometrica}, 84(5):1839--1880, 2016.
\newblock \href {https://doi.org/10.3982/ECTA13322} {\path{doi:10.3982/ECTA13322}}.

\bibitem{brandt2019collective}
Felix Brandt.
\newblock Collective choice lotteries.
\newblock In {\em Studies in Economic Design}, pages 51--73. Springer, 2019.
\newblock \href {https://doi.org/10.1007/978-3-030-18050-8_9} {\path{doi:10.1007/978-3-030-18050-8_9}}.

\bibitem{DBLP:journals/corr/abs-2601-18303}
L{\'{e}}onard Brice, Thomas~A. Henzinger, and K.~S. Thejaswini.
\newblock Dicey games: Shared sources of randomness in distributed systems.
\newblock {\em CoRR}, abs/2601.18303, 2026.
\newblock URL: \url{https://doi.org/10.48550/arXiv.2601.18303}, \href {https://arxiv.org/abs/2601.18303} {\path{arXiv:2601.18303}}, \href {https://doi.org/10.48550/ARXIV.2601.18303} {\path{doi:10.48550/ARXIV.2601.18303}}.

\bibitem{budish2013designing}
Eric Budish, Yeon-Koo Che, Fuhito Kojima, and Paul Milgrom.
\newblock Designing random allocation mechanisms: Theory and applications.
\newblock {\em American Economic Review}, 103(2):585--623, 2013.
\newblock \href {https://doi.org/10.1257/aer.103.2.585} {\path{doi:10.1257/aer.103.2.585}}.

\bibitem{Col74}
George~E. Collins.
\newblock Quantifier elimination for real closed fields by cylindrical algebraic decomposition--preliminary report.
\newblock {\em SIGSAM Bull.}, 8(3):80–90, August 1974.
\newblock \href {https://doi.org/10.1145/1086837.1086852} {\path{doi:10.1145/1086837.1086852}}.

\bibitem{Cournot1838}
Antoine~Augustin Cournot.
\newblock {\em Recherches sur les principes math{\'e}matiques de la th{\'e}orie des richesses}.
\newblock L. Hachette, Paris, 1838.
\newblock English trans: Researches into the Mathematical Principles of the Theory of Wealth (1897).

\bibitem{Zthree}
Leonardo de~Moura and Nikolaj Bjørner.
\newblock Z3: an efficient smt solver.
\newblock volume 4963, pages 337--340, 04 2008.
\newblock \href {https://doi.org/10.1007/978-3-540-78800-3_24} {\path{doi:10.1007/978-3-540-78800-3_24}}.

\bibitem{Gilad2017}
Yossi Gilad, Rotem Hemo, Silvio Micali, Georgios Vlachos, and Nickolai Zeldovich.
\newblock Algorand: {Scaling} byzantine agreements for cryptocurrencies.
\newblock In {\em Proceedings of the 26th Symposium on Operating Systems Principles (SOSP '17)}, pages 51--68. ACM, 2017.
\newblock \href {https://doi.org/10.1145/3132747.3132757} {\path{doi:10.1145/3132747.3132757}}.

\bibitem{Gol10}
Oded Goldreich.
\newblock {\em A Primer on Pseudorandom Generators}, volume~55 of {\em University Lecture Series}.
\newblock American Mathematical Society, 2010.
\newblock \href {https://doi.org/10.1090/ulect/055} {\path{doi:10.1090/ulect/055}}.

\bibitem{hayek1945use}
Friedrich~August Hayek.
\newblock The use of knowledge in society.
\newblock {\em The American economic review}, 35(4):519--530, 1945.

\bibitem{HOWARD1992142}
J.V Howard.
\newblock A social choice rule and its implementation in perfect equilibrium.
\newblock {\em Journal of Economic Theory}, 56(1):142--159, 1992.
\newblock \href {https://doi.org/10.1016/0022-0531(92)90073-Q} {\path{doi:10.1016/0022-0531(92)90073-Q}}.

\bibitem{John48}
Fritz John.
\newblock Extremum problems with inequalities as subsidiary conditions.
\newblock In K.~O. Friedrichs, O.~E. Neugebauer, and J.~J. Stoker, editors, {\em Studies and Essays Presented to R. Courant on his 60th Birthday}, pages 187--204. Interscience Publishers, New York, 1948.

\bibitem{kuhn2001fighting}
Kai-Uwe K{\"u}hn.
\newblock Fighting collusion--regulation of communication between competitors.
\newblock {\em Economic Policy}, 16(32):168--204, 2001.

\bibitem{kwiatkowska2018equilibria}
Marta Kwiatkowska, Gethin Norman, David Parker, and Gabriel Santos.
\newblock Equilibria‐based probabilistic model checking for concurrent stochastic games.
\newblock {\em Formal Methods in System Design}, 53(3):261--294, 2018.

\bibitem{leonard2016geometry}
I.~E. Leonard and J.~E. Lewis.
\newblock {\em Geometry of Convex Sets}.
\newblock Wiley-Blackwell, Hoboken, New Jersey, United States, 1st edition, 2016.

\bibitem{mangasarian67FritzJohn}
Olvi~L. Mangasarian and Stan Fromovitz.
\newblock The fritz john necessary optimality conditions in the presence of equality and inequality constraints.
\newblock {\em Journal of Mathematical Analysis and Applications}, 17(1):37--47, 1967.
\newblock \href {https://doi.org/10.1016/0022-247X(67)90163-1} {\path{doi:10.1016/0022-247X(67)90163-1}}.

\bibitem{smith1982evolution}
John Maynard~Smith and George~R. Price.
\newblock The logic of animal conflict.
\newblock {\em Nature}, 246(5427):15--18, 1973.
\newblock URL: \url{https://royalsocietypublishing.org/doi/10.1098/rstb.2021.0509}.

\bibitem{socialchoiceself2002}
Saptarshi Mukherjee and Hans Peters.
\newblock Self-implementation of social choice correspondences in nash equilibrium.
\newblock {\em Social Choice and Welfare}, 59(4):1009--1028, November 2022.
\newblock URL: \url{https://ideas.repec.org/a/spr/sochwe/v59y2022i4d10.1007_s00355-022-01420-8.html}, \href {https://doi.org/10.1007/s00355-022-01420-8} {\path{doi:10.1007/s00355-022-01420-8}}.

\bibitem{EconNash99}
Roger~B. Myerson.
\newblock Nash equilibrium and the history of economic theory.
\newblock {\em Journal of Economic Literature}, 37(3):1067--1082, September 1999.

\bibitem{Nas50}
John~F. Nash.
\newblock Equilibrium points in $n$-person games.
\newblock {\em Proceedings of the National Academy of Sciences}, 36(1):48--49, 1950.
\newblock \href {https://doi.org/10.1073/pnas.36.1.48} {\path{doi:10.1073/pnas.36.1.48}}.

\bibitem{Schaefer13}
Marcus Schaefer.
\newblock Realizability of graphs and linkages, January 2013.
\newblock URL: \url{https://ideas.repec.org/h/spr/sprchp/978-1-4614-0110-0_24.html}, \href {https://doi.org/10.1007/978-1-4614-0110-0_24} {\path{doi:10.1007/978-1-4614-0110-0_24}}.

\bibitem{Simmons1984}
Gustavus~J. Simmons.
\newblock The prisoners' problem and the subliminal channel.
\newblock In {\em Advances in Cryptology, Proceedings of CRYPTO '83}, volume 196 of {\em Lecture Notes in Computer Science}, pages 51--67. Springer, 1984.
\newblock \href {https://doi.org/10.1007/978-1-4684-4730-9_5} {\path{doi:10.1007/978-1-4684-4730-9_5}}.

\bibitem{Syta2017}
Ewa Syta, Philipp Jovanovic, Eleftherios Kokoris{-}Kogias, Nicolas Gailly, Linus Gasser, Ismail Khoffi, Michael~J. Fischer, and Bryan Ford.
\newblock Scalable bias-resistant distributed randomness.
\newblock In {\em 2017 {IEEE} Symposium on Security and Privacy (SP)}, pages 444--460. IEEE, 2017.
\newblock \href {https://doi.org/10.1109/SP.2017.34} {\path{doi:10.1109/SP.2017.34}}.

\bibitem{Umm10}
Michael Ummels.
\newblock {\em Stochastic multiplayer games: theory and algorithms}.
\newblock PhD thesis, RWTH Aachen, Germany, January 2010.

\bibitem{vives2011information}
Xavier Vives.
\newblock {\em Information and learning in markets: the impact of market microstructure}.
\newblock Princeton University Press, 2010.

\bibitem{vNM44}
John von Neumann and Oskar Morgenstern.
\newblock {\em Theory of Games and Economic Behavior}.
\newblock Princeton University Press, Princeton, NJ, 1944.

\end{thebibliography}

 \appendix
\section{Useful theorems}
We state here some geometry theorems that we use in our proofs.
Note that, in the references cited, the theorem that is proven is often more general: we state it here in a form that fits our purposes.

\begin{theorem}[Carathéodory's theorem~\cite{leonard2016geometry}]
    Let $X \subseteq \Rb^d$ be a set.
    Let $\bx$ be a point of the convex hull of $X$, which we write $\Conv (X)$.
    Then, there exist $d+1$ points $\by_1, \dots, \by_{d+1} \in X$, such that $\bx \in \Conv\{y_1, \dots, y_{d+1}\}$.
\end{theorem}

\begin{theorem}[Fritz John optimality conditions~\cite{John48,mangasarian67FritzJohn}]
    Let $m, n \in \Nb$, and let $f, g_1, \dots, g_n: \Rb^d \to \Rb$ be differentiable functions.
    Let $X = \{\bx \in \Rb \mid g_i(\bx) \geq 0\}$.
    Let $\bx \in X$ be a point such that $\bx = \sup f(X)$.
    
    Let $A = \{i \in [n] \mid g_i(\bx) = 0\}$.
    Then, there exist coefficients $\alpha_0$ and $\alpha_i$ for each $i \in A$, all non-negative and not all zero, such that we have:
    $$\alpha_0 \nabla f(\bx) + \sum_{i \in A} \alpha_i \nabla g_i(\bx) = 0.$$
\end{theorem}
\begin{theorem}[Cell stratification of semialgebraic sets~\cite{BPR06book}]
    Let $S \in \Rb^d$ be a semialgebraic set.
    Let $P_1, \dots, P_n: \Rb^d \to \Rb$ be polynomial functions.
    Then, there exists a partition $S = \bigcup_{\l=1}^n S_\l$ of $S$, where each $S_\l$ is a manifold such that for all $\bx, \by \in S_\l$, the sets $\{i \in [n] \mid P_i(\bx) > 0\}$ and $\{i \in [n] \mid P_i(\by) > 0\}$ are equal, and so are the sets $\{i \in [n] \mid P_i(\bx) < 0\}$ and $\{i \in [n] \mid P_i(\by) < 0\}$.
\end{theorem}

\begin{theorem}[Size of a solution~\cite{BPR94,BPR96,BPR06book}]
    Let $\mathcal{P}$ be a set of polynomials over $\R^d$, each of degree at most $D$, over $N$ variables, with coefficients in $\Zb$ of size at most $L$.
    Let $X = \{\bx \in \Rb^n \mid \forall P \in \mathcal{P}, P(\bx) \geq 0\}$.
    Then, for each connected component $K$ of the set $X$, there exists a point $\bx \in K$ that is algebraic and has size $L D^{O(N)}$.
\end{theorem}\label{app:hammers}

\section{A system of inequations} \label{app:inequations}
Here is the full system of inequations mentioned in \Cref{ssec:triangular}.

\begin{equation}
\begin{split}
    t \leq &~\lambda_1 \lambda_2 \lambda_3 a_{11} b_{ 11} c_{ 11}\\
    &+ \lambda_1 \lambda_2 (1-\lambda_3) a_{11} b_{ 12} c_{ 12}\\
    &+ \lambda_1 (1-\lambda_2) \lambda_3 a_{12} b_{ 21} c_{ 11}\\
    &+ \lambda_1 (1-\lambda_2) (1-\lambda_3) a_{12} b_{ 22} c_{ 12}\\
    &+ (1-\lambda_1) \lambda_2 \lambda_3 a_{21} b_{ 11} c_{ 21}\\
    &+ (1-\lambda_1) \lambda_2 (1-\lambda_3) a_{21} b_{ 12} c_{ 22}\\
    &+ (1-\lambda_1) (1-\lambda_2) \lambda_3 a_{22} b_{ 21} c_{ 21}\\
    &+ (1-\lambda_1) (1-\lambda_2) (1-\lambda_3) a_{22} b_{ 22} c_{ 22}
\end{split}
\end{equation}

\begin{equation}
\begin{split}
    t \leq &~\lambda_1 \lambda_2 \lambda_3 (1-a_{11}) (1-b_{ 11}) (1-c_{ 11})\\
    &+ \lambda_1 \lambda_2 (1-\lambda_3) (1-a_{11}) (1-b_{ 12}) (1-c_{ 12})\\
    &+ \lambda_1 (1-\lambda_2) \lambda_3 (1-a_{12}) (1-b_{ 21}) (1-c_{ 11})\\
    &+ \lambda_1 (1-\lambda_2) (1-\lambda_3) (1-a_{12}) (1-b_{ 22}) (1-c_{ 12})\\
    &+ (1-\lambda_1) \lambda_2 \lambda_3 (1-a_{21}) (1-b_{ 11}) (1-c_{ 21})\\
    &+ (1-\lambda_1) \lambda_2 (1-\lambda_3) (1-a_{21}) (1-b_{ 12}) (1-c_{ 22})\\
    &+ (1-\lambda_1) (1-\lambda_2) \lambda_3 (1-a_{22}) (1-b_{ 21}) (1-c_{ 21})\\
    &+ (1-\lambda_1) (1-\lambda_2) (1-\lambda_3) (1-a_{22}) (1-b_{ 22}) (1-c_{ 22})
\end{split}
\end{equation}
\begin{equation}
    \forall i \in \{1, 2, 3\}, 0 \leq \lambda_i \leq 1
\end{equation}

\section{Proof of \Cref{thm:itsOKtobestraight}} \label{app:itsOKtobestraight}
 \straightLineTheorem*

\begin{proof}
    This proof follows the structure of that of \cref{lm:itsOKtobestraightexample} and generalises it: it first defines a slice-based strategy reshaping transformation, then proves some of its properties, and concludes by using them to transform any strategy into a $k$-grid strategy.

    \paragraph*{The transformation $\phi_\Die$}

    Let $\bstrat$ be a collective strategy, and let $\Die$ be a die.
    For each action $\action$ of the Devil, we define the mapping $f_\action$, which maps each roll $x \in [0,1]$ to the team's expected payoff, when the team follows $\bstrat$, the Devil plays deterministically $\action$, and the outcome of the die $\Die$ is $x$:
    $$f_\action: x \mapsto \Eb(\bstrat, \action \mid \Die = x).$$
    Let us also write $f: x \mapsto (f_\action(x))_{\action \in \Actions_{\opp}}$.
    Let now $V = f([0, 1]) \subseteq \Rb^{\Actions_{\opp}}$ be the set of all values taken by the mapping $f$.
    The integral vector $\int f$ lies in the convex hull $\Conv (V)$.
    By Carathéodory's theorem, since $V$ is contained in a space of dimension $k$, the vector $\int f$ can therefore be written as a convex combination of $k+1$ points of $V$: let us call them $\bv_0, \dots, \bv_k$, and let us write $C$ for the convex hull $C = \Conv\{\bv_0, \dots, \bv_k\}$.
    Then, we have $\int f \in C$.

    Let now $P$ be the polyhedron:
    $$P = \left\{\bv \in \Rb^{\action \in \Actions_{\opp}} ~\left|~ \forall \action, v_\action \geq \int f_\action \right.\right\}$$
    of all vectors that are pointwise greater than or equal to the vector $\int f$.
    The intersection $P \cap C$ is nonempty, since it contains at least the integral $\int f$ itself.
    Therefore, the polyhedron $P$ intersects at least one facet of the polytope $C$, which we can write:
    $$\Conv\{\bv_0, \dots, \bv_{i-1}, \bv_{i+1}, \dots, \bv_k\}.$$
    Let us now pick $x_1, \dots, x_k \in [0,1]$ such that we have $\bv_0 = f(x_1)$, \dots, $\bv_{i-1} = f(x_i)$, $\bv_{i+1} = f(x_{i+1})$, \dots, $\bv_k = f(x_k)$.
    There exist, then, coefficients $\lambda_1, \dots, \lambda_k \in [0,1]$ such that $\sum_{j=1}^k \lambda_j = 1$ and for each action $\action$, we have:
    $$\sum_{j=1}^k \lambda_j f_\action(x_j) \geq \int f_\action.$$

    Let us now define, for each player $\player$, the strategy $\strat'_\player$ as follows.
    If player $\player$ has no access to the die $\Die$, then we have $\strat'_\player = \strat_\player$.
    If player $\player$ has access to the die $\Die$, then we define:
    $$\strat'_\player: \broll \mapsto \begin{cases}
        \strat_\player\left( \broll_{-D}, x_1\right) &\text{if } \roll_\Die \in [0, \lambda_1] \\
        \strat_\player\left( \broll_{-D}, x_2\right) &\text{if } \roll_\Die \in (\lambda_1, \lambda_1+\lambda_2] \\
        \dots \\
        \strat_\player\left( \broll_{-D}, x_k\right) &\text{if } \roll_\Die \in (\lambda_1 + \dots + \lambda_{k-1}, 1].        
    \end{cases}$$

    We then define $\phi_\Die(\bstrat) = \bstrat'$.

\paragraph*{Some properties}
The following proposition is immediate by construction.

\begin{proposition}
    The collective strategy $\bstrat'$ is $k$-piecewise-constant for the die $\Die$.
\end{proposition}

We now prove the following.

\begin{proposition}
    If the collective strategy $\bstrat$ is $k$-piecewise-constant for some die $\Die' \neq \Die$, then so is $\bstrat'$.
\end{proposition}

\begin{proof}
    Let $\broll_{-\Die'}$ be a roll.
    Let us consider the mapping $x \mapsto \bstrat'\left(\broll_{-\Die'}, x\right)$.
    It is equal to the mapping $x \mapsto \bstrat\left(\broll_{-\Die, \Die'}, x_j, x\right)$, with $j$ such that $\roll_\Die \in \left[\lambda_1 + \dots + \lambda_{j-1}, \lambda_1 + \dots + \lambda_j\right]$.
    By hypothesis, that mapping is piecewise constant with $k$ pieces.
\end{proof}

Finally, this transformation maintains or increases the value of the strategy.

\begin{proposition}
    The collective strategy $\bstrat'$ has a value at least as large as the collective strategy $\bstrat$.
\end{proposition}

\begin{proof}
    Let $\action$ be an action for the Devil.
    We need to show that the expected payoff obtained by the collective $\bstrat'$ such an action is at least as large as that obtained by $\bstrat$.
    That is, we need to show:
    $$\Eb\left( \bstrat, \action\right) \leq \Eb\left( \bstrat', \action \right)$$
    $$\text{i.e.}~~~ \int f_\action \leq \int \Eb\left( \bstrat', \action \mid \Die = x \right) \mathsf{d}x$$
    $$= \sum_{j=1}^k \lambda_j \Eb\left( \bstrat, \action \mid \Die = x_j \right)$$
    $$= \sum_{j=1}^k \lambda_j f_\action(x_j)$$
    which is true by definition of the coefficients $\lambda_j$ and rolls $x_j$.
\end{proof}

\paragraph*{Conclusion}

Let us now consider some collective strategy $\bstrat^0$.
By applying successively the transformation $\phi_\Die$ for each die $\Die$, we obtain a $k$-grid collective strategy whose value is at least as large as the one of $\bstrat^0$.    
\end{proof}

\section{Proof of \cref{prop:alpha0neq0}}
\label{app:proofOflinearalgebraProposition}
\alphaneq*
    \begin{proof}[Proof of \cref{prop:alpha0neq0}]
        Assume we have $\alpha_0 = 0$.
        Then, we obtain the equality:
        \begin{equation}\label{eq:alpha0eq0}
            \sum_{i \in I} \alpha_i \nabla P_i(\bblambda, t) = \bzero.
        \end{equation}

        If we first apply the projection $f$ to this equation, we obtain:
        $$\sum_{i \in I} \alpha_i \frac{\partial P_i}{\partial T} (\bblambda, t) = 0.$$
        The partial derivative $\frac{\partial P_i}{\partial T}$
        is constantly equal to $-1$ when $i$ is an action from $\Actions_\dev$, and to $0$ otherwise.
        We then have:
        $$- \sum_{b \in \Actions_\dev} \alpha_b = \bzero.$$
        But since the $\alpha_i$ are all non-negative, this means $\alpha_b = 0$ for each $b$.

        Coming back to \Cref{eq:alpha0eq0}, we can simplify:
        $$\sum_D \left( \alpha_{D\geq} \nabla P_{D\geq}(\bblambda, t) + \alpha_{D\leq} \nabla P_{D\leq}(\bblambda, t) + \sum_{j=1}^k \alpha_{Dj} \nabla P_{Dj}(\bblambda, t)\right) = \bzero.$$

        Let us now compute those gradients.
        For a given die $D$, the gradient $\nabla P_{D\geq}(\bblambda, t)$ consists of $1$'s on all coordinates $(D, j)$ with $j \in [k]$, and $0$ everywhere else.
        The gradient $\nabla P_{D\leq}(\bblambda, t)$ is the opposite.
        For a given index $j$, the gradient $\nabla P_{D j}$ consists of a $1$ on the coordinate $(D, j)$, and $0$ everywhere else.
        Since for each die $D$, the gradients $\nabla P_{D\geq}(\bblambda, t)$, $\nabla P_{D\leq}(\bblambda, t)$, and $\nabla P_{Dj}(\bblambda, t)$ have non-zero coefficients only on coordinates of the form $(D, j')$ with $j' \in [k]$, we can consider each die $D$ separately, and obtain:
        $$\alpha_{D\geq} \nabla P_{D\geq}(\bblambda, t) + \alpha_{D\leq} \nabla P_{D\leq}(\bblambda, t) + \sum_{j=1}^k \alpha_{Dj} \nabla P_{Dj}(\bblambda, t) = \bzero,$$
        which we can write, by ignoring the meaningless coordinates, as:
        $$(\alpha_{D\geq} - \alpha_{D\leq}) (1, \dots, 1) + \alpha_{D1} (1, 0, \dots, 0) + \dots + \alpha_{Dn} (0, \dots, 0, 1) = \bzero.$$
        This is in theory possible with non-zero factors because these vectors are not linearly independent.
        However, since the point $(\bblambda, t)$ satisfies the equality $\sum_j \lambda_{Dj} = 1$, there is at least one index $j$ such that $\lambda_{Dj} > 0$.
        For that index, the constraint $(D, j)$ is therefore not active, which implies $\alpha_{Dj} = 0$.
        Since every strict subset of the set $\{(1, \dots, 1), (1, 0, \dots, 0), \dots, (0, \dots, 0, 1)\}$ is linearly independent, this implies $\alpha_{D\geq} - \alpha_{D\leq} = \alpha_{D1} = \dots = \alpha_{Dn} = 0$, and since the coefficients are non-negative, it also yields $\alpha_{D\geq} = \alpha_{D\leq} = 0$.
        Applying this reasoning for each die $D$, we obtain that all the coefficients $\alpha_i$ are zero, which is excluded by the definition of Fritz John points.
    \end{proof}

\section{Proof of \Cref{lemma:ETRhardInitialProblem}} \label{app:intermediary_hardness_result}
\ETRhardInitialProblem*

\begin{proof}
    By \cite[Lemma~3.9]{Schaefer13}, the following problem is $\exists\Rb$-complete.

    \begin{problem}
        Given a system $S$ that contains:
        \begin{itemize}
            \item equations $E$ of the form:
        $$E: \sum_{i=1}^n \sum_{j=1}^n a_{ij} x_i x_j + \sum_{i=1}^n b_i x_i + c = 0,$$
            \item the inequation $I: \sum_i x_i^2 \leq 1$,
        \end{itemize}
        does $S$ have a solution?
    \end{problem}

    In this problem, the inequality $\sum_i x_i^2 \leq 1$ implies $-1 \leq x_i \leq 1$ for each $i$.
    Each variable $x_i$ can therefore be written $y_i - z_i$, where each $y_i$ and $z_i$ is constrained to lie in the interval $[0, 1]$.
    The system $S$ has therefore a solution if and only if there is a solution to the following system:
    \begin{itemize}
        \item for each $E \in S$, the inequation:
        $$\sum_{i,j} a_{ij} (y_i - z_j)^2 + \sum_i b_i (y_i - z_i) + c \geq 0,$$

        \item for each $E \in S$, the inequation:
        $$\sum_{i,j} (-a_{ij}) (y_i - z_j)^2 + \sum_i (-b_i) (y_i - z_i) + (-c) \geq 0,$$

        \item the inequation $1- \sum_i (y_i - z_j)^2 \geq 0$,

        \item and the inequations $0 \leq y_i \leq 1$ and $0 \leq z_i \leq 1$ for each $i$.
    \end{itemize}
    
    This shows $\exists\Rb$-hardness for the following problem.

    \begin{problem}
        Given a system $S$ that contains:
        \begin{itemize}
            \item inequations $I$ of the form:
        $$I: \sum_{i=1}^n \sum_{j=1}^n a_{ij} x_i x_j + \sum_{i=1}^n b_i x_i + c \leq 0,$$
            \item the inequation $0 \leq x_i \leq 1$ for each $i$,
        \end{itemize}
        does $S$ have a solution?
    \end{problem}

    Now, given such a system, we can introduce a variable $u$ and force it to equal $1$, by adding the inequations $u^2-1 \geq 0$ and $1-u^2 \geq 0$.
    The system $S$ has therefore a solution if and only if there is a solution to the following system:
    \begin{itemize}
        \item for each $I \in S$, the inequation:
        $$\sum_{i,j} a_{ij} x_i x_j + \sum_i b_i x_i u + c \geq 0,$$

        \item the inequations $u^2-1 \geq 0$ and $1-u^2 \geq 0$,

        \item and the inequations $0 \leq x_i \leq 1$ for each $i$, and $0 \leq u \leq 1$.
    \end{itemize}
    
    This shows $\exists\Rb$-hardness for the following problem.

    \begin{problem}
        Given a system $S$ that contains:
        \begin{itemize}
            \item inequations $I$ of the form:
        $$I: \sum_{i=1}^n \sum_{j=1}^n a_{ij} x_i x_j + c \geq 0,$$
            \item the inequation $0 \leq x_i \leq 1$ for each $i$,
        \end{itemize}
        does $S$ have a solution?
    \end{problem}

    Now, given such a system, we know that the sum of the variables $x_i$ cannot exceed $n$.
    We can therefore write $x_i = ny_i$ for each $i \in [n]$, and introduce one more variable $z = 1 - \sum_i y_i$, which is then constrained to also lie in the interval $[0, 1]$.
    The system $S$ has therefore a solution if and only if there is a solution to the following system:
    \begin{itemize}
        \item for each $I \in S$, the inequation:
        $$\sum_{i,j} n^2 a_{ij} y_i y_j + c \geq 0,$$

        \item the inequations $y_i \geq 0$ for each $i$, and $z \geq 0$,

        \item and the equation $\sum_i y_i + z = 1$.
    \end{itemize}

    This shows $\exists\Rb$-hardness for the following problem.

    \begin{problem}
        Given a system $S$ that contains:
        \begin{itemize}
            \item inequations $I$ of the form:
        $$I: \sum_{i=1}^n \sum_{j=1}^n a_{ij} x_i x_j + c \geq 0,$$
            \item the inequation $x_i \geq 0$ for each $i$,
            \item the equation $\sum_{i=1}^n x_i = 1$,
        \end{itemize}
        does $S$ have a solution?
    \end{problem}

    Finally, given such a system, we can use the equality $\sum_{i=1}^n x_i = 1$ to remove the constant term in each inequation $I$.
    Indeed, the system $S$ has a solution if and  only if there is a solution to the following system:
    \begin{itemize}
        \item for each $I \in S$, the inequation:
        $$\sum_{i,j} (a_{ij} + c) x_i x_j \geq 0,$$

        \item the inequations $x_i \geq 0$ for each $i$,

        \item and the equation $\sum_i x_i = 1$.
    \end{itemize}

    This shows $\exists\Rb$-hardness for our initial problem, and concludes the proof.
\end{proof}

\section{Proof of \Cref{lm:allocate_threshold_hard}} \label{app:allocate_threshold_hard}
\allocateThresholdHard*

\begin{proof}
    This hardness proof requires a slight modification of the proof of \Cref{lm:threshold_nexptime_hard}: intuitively, we remove the dice structure from the instance, but we modify the game to force that dice structure to be the only acceptable one.
    From a given formula $\phi$, consider the dicey game $\DGame = (\Game, \Struct)$ defined in that proof.
    We define the game $\Game'$ from the game $\Game$ as follows: we add to the Devil's action set the actions $\rand^2_1$, \dots, $\rand^2_n$, and to each team player's action set the action $2$.
    We modify the payoff function as follows: when the Devil plays the action $\rand^b_i$, with $b \in \{0, 1\}$, and player $x_i$ plays the action $b$, the team payoff is now $3$ (instead of $2$).    
    When the Devil plays the action $\rand^2_i$, the team gets payoff $3$ if player $x_i$ and every player $y_j$ with $x_i \in X_j$ all play the action $2$.
    Otherwise, it gets the payoff $0$.
    In any other case, if any team player plays the action $2$, the team gets payoff $0$. 
    The game is left unchanged under all other aspects.
    
    We still consider the threshold $\threshold = 1$, and define the dice pack $\Pack$ as follows: there are $n$ dice, named $\Die_1, \dots, \Die_n$, and each die $D_i$ has accessibility $\acc(\Die_i) = |\{y_j \mid x_i \in X_j\}| + 1$.

    The dice structure $\Struct$ is compatible with the dice pack $\Pack$.
    Therefore, if the formula $\phi$ is true, we can equip the game $\Game'$ with $\Struct$ and define a collective strategy as in the proof of \Cref{lm:threshold_nexptime_hard}, except that each player $x_i$ is now using the die $\Die_i$ to play the actions $0$, $1$, and $2$ with probability $\frac{1}{3}$ each, and that when player $x_i$ plays the action $2$, so do all players $y_j$ with $x_i \in X_j$ (which is possible since they all have access to the die $\Die_i$).
    It is easy to see that this collective strategy has value $1$.

    Conversely, if there is a dice structure $\Struct'$ and a collective strategy in the dicey game $(\Game', \Struct')$ that guarantees expected payoff at least $1$, then necessarily each player $x_i$ is playing the actions $0$, $1$, and $2$ with probability $\frac{1}{3}$ each.
    Then, when the Devil plays action $\rand^2_i$ for some $i$, the probability that player $x_i$ and the players $y_j$ with $x_i \in X_j$ all play action $2$ (and then get payoff $3$) must be at least $\frac{1}{3}$.
    Since player $x_i$ is already playing action $2$ with probability $\frac{1}{3}$, this means that the probability that all players $y_j$ with $x_i \in X_j$ play action $2$, knowing that player $x_i$ plays action $2$.
    That is only possible if they share one same die.
    Therefore, the dice structure $\Struct'$ is necessarily equal to $\Struct$, up to permutations of dice that have the same accessibility.
    We can then conclude using the same arguments as in the proof of \Cref{lm:threshold_nexptime_hard}.
\end{proof}

\end{document}